\documentclass[%
 amsmath,amssymb,
 aps,
 pre,
 reprint,%
]{revtex4-1}

\usepackage{graphicx}
\usepackage{dcolumn}
\usepackage{bm,epstopdf}
\usepackage[all]{xy}
\usepackage{subfigure}
\usepackage{sidecap,xcolor}
\newtheorem{theorem}{Theorem}[section]
\newtheorem{proof}{Proof}[section]
\usepackage[ansinew]{inputenc} 

\begin{document}

\preprint{AIP/123-QED}

\title[Sample title]{Forecasting the daily and cumulative number of cases for the COVID-19 pandemic in India}
\author{Subhas Khajanchi}%
\email{subhas.maths@presiuniv.ac.in}
\affiliation{Department of Mathematics, Presidency University, 86/1 College Street, Kolkata 700073, India.}

\author{Kankan Sarkar}
 \email{kankan.math@gmail.com}
\affiliation{Department of Mathematics, Malda College, Malda, West Bengal 732101, India
}%

\date{\today}

\begin{abstract}
The ongoing novel coronavirus epidemic has been announced a pandemic by the World Health Organization on March 11, 2020, and the Govt. of India has declared a nationwide lockdown from March 25, 2020, to prevent community transmission of COVID-19. Due to absence of specific antivirals  or vaccine, mathematical modeling play an important role to better understand the disease dynamics and designing strategies to control rapidly spreading infectious diseases. In our study, we developed a new compartmental model that explains the transmission dynamics of COVID-19. We calibrated our proposed model with daily COVID-19 data for the four Indian provinces, namely Jharkhand, Gujarat, Andhra Pradesh, and Chandigarh. We study the qualitative properties of the model including feasible equilibria and their stability with respect to the basic reproduction number $\mathcal{R}_0$.  The disease-free equilibrium becomes stable and the endemic equilibrium becomes unstable when the recovery rate of infected individuals increased but if the disease transmission rate remains higher then the endemic equilibrium always remain stable. For the estimated model parameters, $\mathcal{R}_0 >1$ for all the four provinces, which suggests the significant outbreak of COVID-19.  Short-time prediction shows the increasing trend of daily and cumulative cases of COVID-19 for the four provinces of India.
\end{abstract}

\maketitle

\textbf{In India, 173,763 confirmed cases, 7,964 confirmed new cases and 4,971 confirmed deaths due to COVID-19 were reported as of May 30, 2020. As the ongoing COVID-19 outbreak is quickly spreading throughout India and globe, short-term modeling predictions give time-critical statistics for decisions on containment and mitigation policies. A big problem for short-term prediction is the evaluation of important parameters and how they alter when the first interventions reveal an effect. In the absence of any therapeutics or licensed vaccine, antivirals, isolation of populations diagnosed with COVID-19 and quarantine of populations feared exposed to COVID-19 were used to control the rapid spread of infection. During this alarming situation, forecasting is of utmost priority for health care planning and to control the SARS-CoV-2 virus with limited resources. We proposed a mathematical model that monitors the dynamics of six compartments, namely susceptible (S), asymptomatic (A), reported symptomatic (I), unreported symptomatic (U), quarantine (Q), and recovered (R) individuals, collectively termed SAIUQR, that predicts the course of the epidemic. Our SAIUQR model discriminates between reported and unreported infected individuals, which is important as the former are typically isolated and hence less likely to spread the infection. A detailed theoretical analysis has been done for our SAIUQR model in terms of the basic reproduction number $\mathcal{R}_0$. All the analytical findings are verified numerically for the estimated model parameters. We have calibrated our SAIUQR model with real observed data on the COVID-19 outbreak in the four provinces of India. The basic reproduction number for all the provinces is greater than unity, which resulted in a substantial outbreak of COVID-19. Based on the simulation, our SAIUQR model predict that on June 13, 2020, the daily new COVID-19 cases will be around 15, 454, 12, 96, and cumulative number of COVID-19 cases will be around 661, 23955, 514, 4487, in Jharkhand, Gujarat, Chandigarh and Andhra Pradesh, respectively.}

\section{Introduction}
\label{Intro}

After a novel strain of COVID-19, was detected in Wuhan,  the city of Hubei province, China, in December 2019 \cite{Wu20}, an exponentially increasing number of patients in mainland China were identified with SARS-CoV-2, immediately the Chinese Health authorities to initiate radical measures to control the epidemic coronavirus. In spite of these radical measures, a SARS-CoV-2 coronavirus pandemic ensued in the subsequent months and China became the epicenter. SARS-CoV-2 viruses are enveloped non-segmented positive-sense RNA viruses that belong to the Coronaviridae family and the order Nidovirales, and are extensively disseminated among humans as well as mammals \cite{Chen20}. The COVID-19 is responsible for a range of symptoms together with fever, dry cough, breathing difficulties, fatigue and lung infiltration in severe cases, similar to those created by SARS-CoV (severe acute respiratory syndrome coronavirus) and MERS-CoV (Middle East respiratory syndrome coronavirus) infections \cite{Huang20}. SARS-CoV-2 has already crossed the earlier history of two coronavirus epidemics SARS-CoV and MERS-CoV, posing the substantial threat to the world people health problem as well as economical problem after the second world war \cite{bbc1}.
According to the World Health Organization report, dated 29 May 2020 reported 5,596,550 total cases and 353,373 deaths worldwide \cite{WHO129}.

Till date, there is no licensed vaccine, drugs and effective therapeutics available for SARS-CoV-2 or COVID-19. Due to absence of pharmaceutical interventions, Govt. of various countries are adopting different strategies to control the outbreak and the most common one is the nation-wide lockdown. It was started with the local Govt. of Wuhan by temporarily prevention of all public traffics within the city on January 23, 2020 and soon followed by other cities in Hubei province \cite{Lina20}. In the absence of drug or specific antivirals for SARS-CoV-2 viruses, maintaining social distancing is the only way to mitigate the human-to-human transmission for coronavirus diseases, and thus the other countries also incorporated the strict lockdown, quarantines and curfews.

In India, the first coronavirus case was reported in Kerala's Thrissur district on January 30, 2020 when a student returned back from Wuhan, the sprawling capital of China's Hubei province \cite{IndCov}. The Govt. of India has implemented a complete nationwide lockdown throughout the country on and from March 25, 2020 for 21 days, following one day `Janata Curfew' on March 22, 2020 to control the coronavirus or SARS-CoV-2 pandemic in India \cite{Pulla20}. Due to massive spread of coronavirus diseases, the Govt. of India has extended the lockdown and it is going on Phase 4; from May 18, 2020 to May 31, 2020. Besides the implementation of nationwide lockdown, the Ministry of Health and Family Welfare (MOHFW) of India, recommended different individual hygiene measures, for example, frequent hand washing, social distancing, use of mask, avoid gathering and touching eyes, mouth and nose etc. \cite{mohfw}.

The Govt. also ceaselessly using different media and social network to aware the public regarding coronavirus diseases and its precautions. Albeit, the factors such as diverse and huge population, the unavailability of specific therapeutics, drugs or licensed vaccines, inadequate evidences regarding the mechanism of disease transmission make it strenuous to combat against the coronavirus diseases throughout India. To control the transmission of COVID-19, lockdown is a magnificent measure but testing is also an important factor to identify the symptomatic and asymptomatic individuals. The symptomatic individuals should be reported by the public health agencies to separate them from the uninfected or asymptomatic individuals for their ICU (Intensive Care Unit) treatment. Also, from an economic viewpoint, the strict lockdown may be the cause of a substantial financial crisis in near future. In particular, the lockdown in high dense countries can mitigate the disease transmission rate, although entirely control not be obtainable. Thus, to survive the economical status of a country, a strict lockdown for a long period is not advisable at all in any situations. Hence there should be a acceptable balance between the two different characteristics of governmental strategies: strict lockdown and healthy economical situations. Albeit, few questions remains to answer whether this cluster containment policy can be effective in mitigating SARS-CoV-2 transmission or not ? If not then what can be the possible solutions to mitigate the transmission of SARS-CoV-2 viruses ? These questions can only be replied by investigating the dynamics and forecasting of a mechanistic compartmental model for SARS-CoV-2 transmission and comparing the outcomes with real scenarios.

A plenty of mathematical models has been investigated to study the transmission dynamics and forecasting of COVID-19 outbreak \cite{Kucharski20,Tang20,WuJT20,Fanelli20,Ribeiro20,Chakraborty20,KhCov20,Sarkar20,He20,Anastassopoulou,Giordano20}. Kucharski and colleagues \cite{Kucharski20} performed a model-based analysis for SARS-CoV-2 viruses and calculate the reproduction number $\mathcal{R}_0 = 2.35$, where the authors have taken into account all the positive cases of Wuhan, China, till March 05, 2020.  Wu and colleagues \cite{WuJT20} studied a susceptible-exposed-infectious-recovered (SEIR) model to simulate the epidemic in Wuhan city and compute the basic reproduction number $\mathcal{R}_0 = 2.68$ and predict their model based on the data recorded from December 31, 2019 to January 28, 2020. Tang and colleagues \cite{Tang20} developed a compartmental model to study the transmission dynamics of COVID-19 and calculate the basic reproduction number $\mathcal{R}_0 = 6.47$, which is very high for the infectious diseases. Recently, Fanelli \& Piazza \cite{Fanelli20} analyzed and predicted the characteristics of SARS-CoV-2 viruses in the three mostly affected countries till March 2020 with an aid of the mathematical modeling. Stochastic based regression model also been studied by Ribeiro and Colleagues \cite{Ribeiro20} to predict the scenarios of the most affected states of Brazil. Chakraborty \& Ghosh \cite{Chakraborty20} investigate a hybrid ARIMA-WBF model to predict the various SARS-CoV-2 affected countries throughout the world. Khajanchi and Colleagues \cite{KhCov20} developed a compartmental model to forecast and control of the outbreak of COVID-19 in the four states of India and the overall India. Sarkar \& Khajanchi \cite{Sarkar20} developed a mathematical model to study the model dynamics and forecast the SARS-CoV-2 viruses in seventeenth provinces of India and the overall India. A discrete-time SIR model introducing dead compartment system studied by Anastassopoulou et al. \cite{Anastassopoulou} to portray the dynamics of SARS-CoV-2 outbreak. Giordano and colleagues \cite{Giordano20} established a new mathematical model for COVID-19 pandemic and predict that restrictive social distancing can mitigate the widespread of COVID-19 among the human. A couple of seminal papers has been investigated to study the transmission dynamics of COVID-19 or SARS-CoV-2 viruses for different countries, including Mexico city, Chicago and Wuhan, the sprawling capital of Central China's Hubei province \cite{Xiao-Lin20,Ndairou20,Mena20,Wong20}. Short-term prediction is too important as it gives time-critical information for decisions on containment and mitigation strategies \cite{KhCov20,Dehning20}. A major problem for short-term predictions is the evaluation of important epidemiological parameters and how they alter when first interventions reveal an effect.

The main objective of this work is to develop a new mathematical model that describes the transmission dynamics and forecasting of COVID-19 or SARS-CoV-2 pandemic in the four different provinces of India, namely Jharkahnd, Andhra Pradesh, Chandigarh, and Gujarat. We estimated the model parameters of four different states of India and fitted our compartmental model to the daily confirmed cases and cumulative confirmed cases reported between March 15, 2020 to May 24, 2020. We compute the basic reproduction number $\mathcal{R}_0$ for the four different states based on the estimated parameter values. We also perform the short-term predictions of the four different states of India from May 25, 2020 to June 13, 2020, and it shows the increasing trends of COVID-19 pandemic in four different provinces of India.

The remaining part of this manuscript has been organized in the following way. In the Section 2, we describe the formulation of the compartmental model for COVID-19 and its basic assumptions. Section 3 describes the theoretical analysis of the model, which incorporates the positivity and boundedness of the system, computation of the basic reproduction number $\mathcal{R}_0$, existence of the biologically feasible singular points and their local stability analysis. In the same section, we perform the global stability analysis for the infection free equilibrium point $E^0$ and the existence of transcritical bifurcation at the threshold $\mathcal{R}_0 = 1.$ In the Section 4, we conduct some model simulations to validate our analytical findings by using the estimated model parameters for Jharkhand, the state of India. The parameters are estimated for the real world example on COVID-19 for four different states of India and perform a short-term prediction based on the estimated parameter values. A discussion in the Section 5 concludes the manuscript.

\section{Mathematical model}
\label{sec:1}

A compartmental mathematical model has been developed to study the transmission dynamics of COVID-19 outbreak in India and throughout the world. We adopt a variant that focuses some important epidemiological properties of COVID-19 or SARS-CoV-2 coronavirus diseases. Based on the health status, we stratify the total human population into six compartments, namely susceptible or uninfected $(S)$, asymptomatic or pauci-symptomatic infected $(A)$, symptomatic reported infected $(I)$, unreported infected $(U)$, quarantine $(Q)$, and recovered $(R)$ individuals, collectively termed SAIUQR. At any instant of time, the total population is denoted by $N = S + A + I + U + Q + R.$ Depending on the six state variables, we aim to develop an autonomous system using first order nonlinear ordinary differential equations.

In the model formulation, quarantine refers to the separation of coronavirus infected population from the general population when the individuals are infected but clinical symptoms has not yet developed, whereas isolation refers to the separation of coronavirus infected population when the population already identified the clinical symptoms. Our mathematical model introduces some demographic effects by assuming a proportional natural mortality rate $\delta > 0$ in each of the six compartments. In addition, our model incorporates a constant recruitment of susceptible populations into the region at the rate $\Lambda_s$ per unit time. This parameter represents new birth, immigration and emigration. The parameter $\beta_s$ represents the probability of disease transmission rate. However the disease transmission from vulnerable to infected individuals (for our model, the class is A) depend on various factors, namely safeguard precautions (use of mask, social distancing, etc.) and hygienic safeguard (use of hand sanitizer) taken by the susceptible individuals as well as infected population. In our model formulation, we incorporate the asymptomatic or pauci-symptomatic infected (undetected) individuals, which is important to better understand the transmission dynamics of COVID-19, which also studied by Giordano et al. \cite{Giordano20} and Xiao-Lin et al. \cite{Xiao-Lin20}.

In our model formulation, we assumed that the COVID-19 virus is spreading when a vulnerable person come into contact with an asymptomatic infected individuals. The uninfected individuals decreases after infection, obtained through interplays among a susceptible population and an infected individuals who may be asymptomatic, reported symptomatic and unreported symptomatic. For these three compartments of infected population, the transmission coefficients are $\beta_s\alpha_a$, $\beta_s\alpha_i$, and $\beta_s\alpha_u$ respectively. We consider $\beta_s$ as the disease transmission rate along with the adjustment factors for asymptomatic $\alpha_a$, reported symptomatic $\alpha_i$ and unreported symptomatic $\alpha_u$ individuals. The interplays among infected populations (asymptomatic, reported symptomatic, and unreported symptomatic) and susceptible individuals can be modeled in the form of total individuals using standard mixing incidence \cite{Anderson91,Diekmann00,Hethcote00,Gumel04}.

The quarantined population can either move to the susceptible or infected compartment (reported and unreported), depending on whether they are infected or not \cite{Keeling08}, with a portion $\rho_s$. Here, $\gamma_q$ is the rate at which the quarantined uninfected contacts were released into the wider community. Asymptomatic individuals were exposed to the virus but clinical symptoms of SARS-CoV-2 viruses has not yet developed. The asymptomatic individuals decreases due to contact with reported and unreported symptomatic individuals at the rate $\gamma_a$ with a portion $\theta \in (0, 1)$, and become quarantine at the rate $\xi_a$. Also, the asymptomatic individuals become recovered at the rate $\eta_a$ and has a natural mortality rate $\delta$. A fraction of quarantine individuals become reported infected individuals at the rate $\gamma_q$ with a portion $\rho_s$ (where $\rho_s \in (0, 1))$.

As we know that an individual is whether infected by coronavirus diseases or not can be identified by RT-PCR screening test and a person with negative results with the RT-PCR screening test may yet be coronavirus positive as it may take around 7-21 days to express the coronavirus symptoms \cite{Lan20}. Thus, a fraction of coronavirus positive class can be considered as reported symptomatic individuals $(\theta)$ and unreported symptomatic individuals $(1-\theta)$.  The reported symptomatic individuals separated from the general populations and move to the isolated class or hospitalized class for clinical treatment.

Also, it can be noticed that once an individual recovered from the SARS-CoV-2 diseases has a very little chance to become infected again for the same disease \cite{Lan20}. Therefore, we assume that none of the recovered individuals move to the susceptible or uninfected class again. In our mathematical model formulation, we assume that the reported infected individuals $(I)$ are unable to spread or transmit the viruses as they are kept completely isolated from the susceptible or uninfected individuals. As the reported infected individuals are moved to the hospital or Intensive Care Unit (ICU) \cite{Liu20}. For our modeling perspective, we are mainly interested in predictions over a relatively short time window within which the temporary immunity is likely still to be in placed, and the possibility of reinfection would negligibly affect the total number of uninfected populations and so there would be no considerable difference in the evolution of the epidemic curves we consider. Social mixing patterns are introduced into our contagion parameters in an average fashion over the entire individuals, irrespective of age. Based on these biological assumptions, we develop the following mathematical model using a system of nonlinear ordinary differential equations to study the outbreak of COVID-19 or SARS-CoV-2 coronavirus diseases:
\begin{widetext}
\begin{eqnarray}
\label{stateeq}
\left\{
  \begin{array}{ll}
    S'(t) = \Lambda_{s} - \beta_s S \bigg( \alpha_a \frac{A}{N} + \alpha_i \frac{I}{N} + \alpha_u \frac{U}{N} \bigg) + \rho_{s} \gamma_q Q - \delta S,  \\
    A'(t) = \beta_s S \bigg( \alpha_a \frac{A}{N} + \alpha_i \frac{I}{N} + \alpha_u \frac{U}{N} \bigg) - (\xi_a + \gamma_a) A - \eta_a A - \delta A, \\
    I'(t) = \theta \gamma_a A + (1-\rho_s) \gamma_q Q - \eta_i I - \delta I, \\
    U'(t) = (1 - \theta) \gamma_{a} A - \eta_u U - \delta U, \\
    Q'(t) = \xi_a A - \gamma_q Q - \delta Q,  \\
    R'(t) = \eta_u U + \eta_{i} I + \eta_{a} A - \delta R,
     \end{array}
\right.
\end{eqnarray}
\end{widetext}

\noindent the model is supplemented by the following initial values:
\begin{eqnarray}
\label{IC}
\nonumber
S(t_0) &=& S_0 \geq 0,~A(t_0) = A_{0} \geq 0,~Q(t_0) = Q_0 \geq 0,\\
I(t_0) &=& I_0 \geq 0,~U(t_0) = U_{0} \geq 0,~R(t_0) = R_0 \geq 0.
\end{eqnarray}

In our model $t \geq t_0$ is the time in days, $t_0$ represents the starting date of the outbreak for our system (\ref{stateeq}).  The transmission dynamics of the COVID-19 is illustrated in the Figure \ref{Schema}. The description of the model parameters are presented in the Table \ref{parval}.


\begin{figure}
\includegraphics[width=0.45\textwidth]{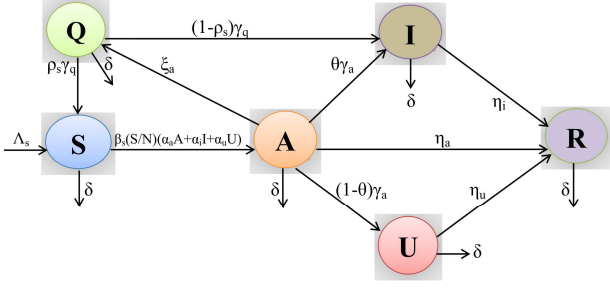}
\caption{ A schematic representation of the mechanistic SAIUQR model for the transmission dynamics of COVID-19 or SARS-CoV-2. The interaction among different stages of individuals is shown in the graphical scheme: S, susceptible or uninfected population; A, asymptomatic infected population; I, COVID-19 reported symptomatic infected individuals; U, COVID-19 unreported symptomatic infected individuals; Q, quarantine individuals; and R, COVID-19 recovered individuals. Biological interpretations of the model parameters are given in the Table \ref{parval}.}
\label{Schema}
\end{figure}

\section{SAIUQR model analysis}

In this section, we provide the basic properties of the SAIUQR model (\ref{stateeq}), including positivity and boundedness of the solutions, basic reproduction number and the biologically feasible singular points and their stability analysis, subject to the non-negative initial values $(S_0, A_0, Q_0, I_0, U_0, Q_0, R_0) \in \mathbf{R}^{6}.$

\begin{theorem}
\label{Postivity}	
The solutions of the SAIUQR system (\ref{stateeq}) with the initial values (\ref{IC}) are defined on $\mathbf{R}_{+}^{6}$ remain positive for all $t > 0$.
\end{theorem}
\begin{proof}
The proof of this theorem is given in the Appendix \ref{PositivityProof}.
\end{proof}

\begin{theorem}
\label{Bounded}	
The solutions of the SAIUQR system (\ref{stateeq}) with the initial conditions (\ref{IC}) are uniformly bounded in the region $\Omega$.
\end{theorem}
\begin{proof}
The proof of this theorem is given in the Appendix \ref{BoundedProof}.
\end{proof}

\subsection{Basic reproduction number}

In any infectious disease modeling, the basic reproduction number is the key epidemiological parameter for describing the characteristics of the diseases. The basic reproduction number symbolized by $\mathcal{R}_{0}$ and is defined as ``the number of secondary infected individuals caused by a single infected individuals in the entire susceptible individuals" \cite{Diekmann90}. The dimensionless basic reproduction number $\mathcal{R}_{0}$ quantifies the expectation of the disease die out or the spreading of the diseases. For, $\mathcal{R}_{0} < 1$ describes on an average an infected population spreads less than a new infective during the course of its infection period, and thus the diseases can die out.  For, $\mathcal{R}_{0} > 1$ describes each infected individuals spread on an average more than 1 new infection, and the disease can spread throughout the population. Various techniques can be used to compute the basic reproduction number $\mathcal{R}_{0}$ for an epidemic outbreak. In our present study, we use the next generation matrix to evaluate $\mathcal{R}_{0}$ \cite{Diekmann90}. In our compartmental model, the following classes are explicitly related to the outbreak of the novel coronavirus diseases: $A$, $I$, $U$, and $Q$. Thus, from the SAIUQR model system (\ref{stateeq}), we get the matrices $\mathcal{F}$ for the new infection and $\mathcal{V}$ for the transition terms are given by, respectively

\begin{eqnarray*}
\mathcal{F} &=& \left[ \begin{array}{c}
       \beta_s S \bigg(\alpha_a \frac{A}{N} + \alpha_i \frac{I}{N} + \alpha_u \frac{U}{N} \bigg)  \\
       0 \\
       0 \\
       0
     \end{array} \right],\\
     \mathcal{V} &=& \left[ \begin{array}{c}
       (\xi_a + \gamma_a + \eta_a + \delta)A  \\
       - \theta \gamma_a A  - (1-\rho_s)\gamma_q Q + (\eta_i + \delta) I \\
      - (1 - \theta) \gamma_a A + (\eta_u + \delta) U   \\
        - \xi_a A + (\gamma_q + \delta) Q
     \end{array} \right].
\end{eqnarray*}

The variational matrix for the SAIUQR system (\ref{stateeq}), can be evaluated at an infection-free singular point $E^{0}(S^{0}, A^{0}, Q^{0}, I^{0}, U^{0}, R^{0}) = (\frac{\Lambda_{s}}{\delta}, 0, 0, 0, 0, 0)$,  we have
\begin{eqnarray*}
F &=& \left[ \begin{array}{cccc}
       \beta_s \alpha_a \frac{S}{N} & \beta_s \alpha_i \frac{S}{N} & \beta_s \alpha_u \frac{S}{N} & 0 \\
       0 & 0 & 0 & 0 \\
       0 & 0 & 0 & 0 \\
       0 & 0 & 0 & 0
     \end{array} \right],\\
V &=& \left[ \begin{array}{cccc}
       \xi_a + \gamma_a + \eta_a + \delta & 0 & 0 & 0 \\
       - \theta \gamma_a & (\eta_i + \delta) & 0 & - (1-\rho_s)\gamma_q \\
       - (1-\theta) \gamma_a & 0 & (\eta_u + \delta) & 0 \\
       - \xi_a & 0 & 0 & \gamma_q + \delta
     \end{array} \right].
\end{eqnarray*}

The basic reproduction number $\mathcal{R}_0 = \rho(FV^{-1})$, where $\rho(FV^{-1})$ represents the spectral radius for a next generation matrix $FV^{-1}$. Therefore, from the SAIUQR model (\ref{stateeq}), we get the basic reproduction number $\mathcal{R}_0$ as
\begin{eqnarray*}
\label{R0}
\rho(F V^{-1}) &=& \mathcal{R}_{0} ~=~ \frac{\beta_s \alpha_a}{\xi_a + \gamma_a + \eta_a + \delta} \nonumber\\
 &+& ~~ \frac{(1-\theta)\beta_s \alpha_u \gamma_a}{(\eta_u + \delta)(\xi_a + \gamma_a +\eta_u + \delta)}  \nonumber\\[0.2cm]
&+& \frac{\beta_s \alpha_i (\theta \gamma_a (\gamma_q +\delta) + (1-\rho_s)\xi_a\gamma_q)}{(\xi_a + \gamma_a + \eta_a + \delta)(\gamma_q + \delta) (\eta_i +\delta)}.
\end{eqnarray*}

\subsection{Equilibria}

The SAIUQR model (\ref{stateeq}) has two biologically feasible equilibrium points, namely   \\
(i) infection free steady state $E^{0}(S^{0}, A^{0}, Q^{0}, I^{0}, U^{0}, R^{0}) = (\frac{\Lambda_{s}}{\delta}, 0, 0, 0, 0, 0)$, and    \\
(ii) the endemic equilibrium point $E^{*}(S^{*}, A^{*}, Q^{*}, I^{*}, U^{*}, R^{*})$, where   \\
$S^*=\frac{\Lambda_s}{\delta}- \hat{S}A^*,~I^*=\hat{I}A^*,~ U^*=\hat{U}A^*,~ Q^*=\hat{Q}A^*$ and $R^*=\hat{R}A^*$. The expression of $A^*$ is given by $A^*=\frac{\Lambda_s(\mathcal{R}_{0} -1)}{\delta\left\lbrace \hat{S}(\mathcal{R}_{0}-1)+1+\hat{I}+\hat{U}+\hat{Q}+\hat{R} \right\rbrace}$, where $\hat{S}=\frac{1}{\delta}\left[\xi_a+\gamma_a+\eta_a+\delta-\frac{\xi_a\rho_s\gamma_q}{\gamma_q+\delta} \right]$,

\noindent $\hat{I}=\frac{\theta\gamma_a(\gamma_q+\delta)+(1-\rho_s)\xi_a\gamma_q}{(\eta_i+\delta)(\gamma_q+\delta)}$, $\hat{U}=\frac{(1-\theta)\gamma_a}{\eta_u+\delta}$, $\hat{Q}=\frac{\xi_a}{\gamma_q+\delta} $
\noindent  and $\hat{R}=\frac{1}{\delta}\bigg[ \frac{(1-\theta)\gamma_a\eta_u}{\eta_u+\delta} + \frac{\eta_i\left\lbrace \theta\gamma_q(\gamma_q+\delta)+(1-\rho_s)\gamma_q\xi_a\right\rbrace }{(\gamma_q+\delta)(\eta_i+\delta)} + \eta_a \bigg]  $.

It can be observed the the infection free singular point $E^{0}(S^{0}, A^{0}, Q^{0}, I^{0}, U^{0}, R^{0})$ is always feasible and the endemic equilibrium point $E^{*}(S^{*}, A^{*}, Q^{*}, I^{*}, U^{*}, R^{*})$ is feasible if the following condition holds:
\begin{eqnarray*}
(i)&&\mathcal{R}_0 ~>~ 1,  \\
(ii)&& \frac{1}{\xi_a+\gamma_a+\eta_a+\delta} \bigg[ \frac{\Lambda_{s}}{A^*} + \frac{\xi_a \rho_s \gamma_q}{\gamma_q +\delta}  \bigg] ~>~ 1.
\end{eqnarray*}

\subsection{Stability analysis}

In the present subsection, we investigate the linear stability analysis for the SAIUQR model (\ref{stateeq}) for the two feasible steady states. By using the techniques of linearization, we investigate the local dynamics of the complicated system of the coronavirus compartmental model. Generally, we linearize the SAIUQR model around each of the feasible steady states and perturb the compartmental model by a very small amount, and observe whether the compartmental model returns to that steady states or converges to any other steady states or attractor. The local stability analysis aids in understanding the qualitative behavior of the complex nonlinear dynamical system. By using the following theorem, we prove the local stability of the infection free singular point $E^{0}(S^{0}, A^{0}, Q^{0}, I^{0}, U^{0}, R^{0})$:

\begin{theorem}
\label{LASE0}
The infection free steady state  $E^{0}$ is locally asymptotic stable if $\mathcal{R}_0 < 1$ and unstable if $\mathcal{R}_0 > 1$.
\end{theorem}

\begin{proof}
The proof of this theorem is given in Appendix \ref{LASE0Proof}.
\end{proof}

\begin{theorem}
\label{GASE0}
The infection free steady state  $E^{0}$ is globally asymptotic stable for $\mathcal{R}_0 < 1$ in the bounded region $\Omega$.
\end{theorem}

\begin{proof}
The proof of this theorem is given in Appendix \ref{GASE0Proof}.
\end{proof}

\begin{theorem}
\label{trans-EEs}
The SAIUQR model system (\ref{stateeq}) admits a locally asymptotic stable around the endemic equilibrium point $E^*$ for $\mathcal{R}_0 > 1$. Also, the system (\ref{stateeq}) experiences forward bifurcation at $\mathcal{R}_0 = 1$.
\end{theorem}

\begin{proof}
	The proof of this theorem is given in Appendix \ref{trans-EEsProof}.
\end{proof}

\begin{table*}[ht]
  \caption{Table of the biologically relevant parameter values and their description for the SAIUQR model system (\ref{stateeq}).}
\label{parval}
 {\begin{tabular}{p{1.3cm} p{9.8cm} l l }
  \\[-0.1cm]
  \hline
  Symbol      & Biological interpretations & Values \& Source   \\[0.1cm]
  \hline
 $\Lambda_s $ & birth rate of the susceptible individuals & Table ~\ref{parvalIC} \\[0.18cm]
  $\beta_s$   & probability of the disease transmission coefficient &  Estimated \\[0.18cm]
  $\alpha_a$  & modification factor for asymptomatic infected individuals   & Estimated \\[0.18cm]
  $\alpha_i$  & modification factor for symptomatic infected individuals   & Estimated \\[0.18cm]
  $\alpha_u$  & modification factor for unreported infected individuals  & Estimated  \\[0.18cm]
  $\rho_s$    & fraction of quarantine individuals that become susceptible individuals & 0.5~(0, ~1)~ Fixed \\[0.18cm]
  $\gamma_q$  & rate at which the quarantined individuals becomes susceptible individuals & Estimated \\[0.18cm]
  $\delta $   & natural death rate of entire individuals & 0.1945 $\times 10^{-4}$ ~ \cite{KhCov20} \\[0.18cm]
  $\xi_a$     & rate at which the asymptomatic individuals become quarantined & 0.07151 ~ \cite{Sarkar20} \\[0.18cm]
  $\gamma_a $ & rate of transition from the asymptomatic individuals to infected individuals &  Estimated  \\[0.18cm]
  $\eta_a $   & average recovery rate of asymptomatic individuals & $\frac{1}{7.48}$ ~ \cite{Sarkar20} \\[0.18cm]
  $\theta $   & fraction of asymptomatic individuals that become reported infected individuals & 0.8~(0, ~1) ~ Fixed \\[0.18cm]
  $\eta_i$    & average recovery rate of reported symptomatic infected individuals & $\frac{1}{7}$ ~ \cite{Liu20} \\[0.18cm]
  $\eta_u$    & average recovery rate of unreported symptomatic infected individuals & $\frac{1}{7}$ ~ \cite{Liu20} \\[0.15cm]
  \hline
  \end{tabular}}
\end{table*}

\begin{figure*}[ht!]
\centering
\includegraphics[width=\textwidth]{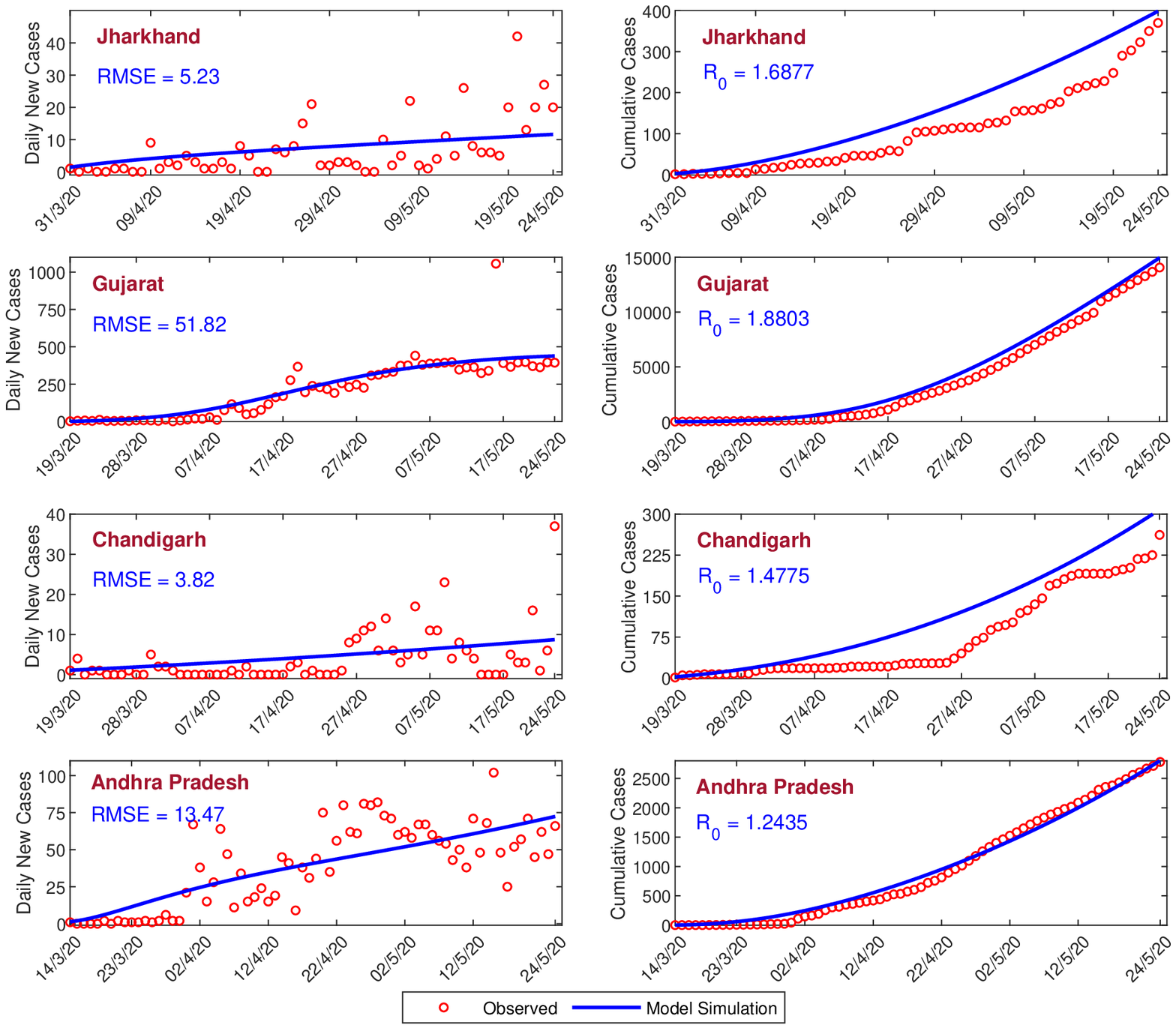}
\caption{Model estimation based on the observed data. Model simulations fitted with the daily new cases and the cumulative confirmed cases of COVID-19 for four provinces of India, namely Jharkhand, Gujarat, Chandigarh, and Andhra Pradesh. Observed data points are shown in the red circle while the blue curve represents the best fitting curve for the SAIUQR model. The first column represents the daily new cases and the second column represents the cumulative confirmed cases of COVID-19. The estimated parameter values are listed in the Table \ref{parvalEst} and other parameter values are listed in the Table \ref{parval}. The initial values used for this parameter values are presented in the Table \ref{parvalIC}. The RMSE and the value of $R_0$ for each provinces are mentioned in the inset.}
\label{F:DataFit}
\end{figure*}

\section{Numerical simulation}

In this section, we conduct some numerical illustrations to validate our analytical findings. Analytically, we perform the local stability analysis for infection free steady state $E^0$ and a unique endemic equilibrium point $E^*$. We also perform the transcritical bifurcation at the threshold $\mathcal{R}_{0} = 1$ and the global stability analysis for disease-free steady state $E^0$. In order to validate the analytical calculations, we used the estimated parameter values for Jharkhand, the state of India and the techniques for parameter estimation are described in the subsection \ref{modcalb}.

\subsection{Model calibration}
\label{modcalb}
We have calibrated our SAIUQR model (\ref{stateeq}) with the observed daily new COVID-19 cases. We have considered three states of India namely Jharkahnd, Gujarat, Andhra Pradesh and one city of India namely Chandigarh. The daily new COVID-19 cases are collected from the first COVID-19 case reported and up to May 24, 2020. The daily reported COVID-19 data were obtained from COVID19 INDIA (https://www.covid19india.org/) \cite{covid-19tr}. We have estimated six model parameters, namely  $\beta_s$, $\alpha_a$, $\alpha_i$, $\alpha_u$,  $\gamma_a$ and $\gamma_q$, out of fourteen system parameters for the system (\ref{stateeq}) by using least square method \cite{Banerjee15}. The values of these parameters and the initial population size plays an important role in the model simulation. The parameters are estimated by assuming the initial population size. The initial population are presented in the Table \ref{parvalIC}. Three days moving average filter has been applied to the daily COVID-19 cases to smooth the data. The daily reported confirmed COVID-19 cases are fitted with the model simulation by using the least square method. The estimated parameter values are listed in the Table \ref{parvalEst}. Different set of parameter values locally minimizes the Root Mean Square Error (RMSE) and we have considered the set of parameter values, which gives the realistic value of the basic reproduction number  $\mathcal{R}_{0}$. RMSE is the measure of the accuracy of the fitting data and the RMSE is defined as follows:

\begin{eqnarray}
\begin{aligned}
RMSE=\sqrt{\frac{\Sigma_{i=1}^{n} ( O(i)-M(i))^2}{n}},
\nonumber
\end{aligned}
\end{eqnarray}

\noindent where $n$ represents the size of the observed data, $O(i)$ is the reported daily confirmed COVID-19 cases and $M(i)$ represents the model simulation. The Figure \ref{F:DataFit} shows the daily confirmed COVID-19 cases (first column), cumulative confirmed COVID-19 cases (second column) and model simulations has been shown in the blue curve for all the four provinces of India. The values of RMSE and basic reproduction number $\mathcal{R}_{0}$ for all the four provinces are presented in the inset of the figure. The SAIUQR model performs well for the three provinces, namely Jharkhand, Chandigarh, and Andhra Pradesh. The RMSE for Gujarat is higher than the other provinces as the number of daily confirmed COVID-19 cases are higher than the other provinces. The value of the basic reproduction number $\mathcal{R}_{0}$ for Jharkhand, Gujarat, Chandigarh, and Andhra Pradesh are 1.6877, 1.8803, 1.4775, and 1.2435, respectively and the trend of daily confirmed COVID-19 cases is increasing. This increasing trend of the daily new COVID-19 cases for all the four provinces of India are captured by our model simulation. In all the four provinces $\mathcal{R}_{0} > 1$, so the disease free equilibrium point $E^0$ is unstable. The basic reproduction number for the four states are greater than unity, which indicates the substantial outbreak of the COVID-19 in the states.

\subsection{Validation of analytical findings}

In this section, we have validated our analytically findings by using numerical simulations for the parameter values in the Table \ref{parval}, and the estimated parameter values in the Table \ref{parvalEst} for our SAIUQR model for the coronavirus diseases. The parameter values are estimated for the observed COVID-19 data for the three states of India, namely Jharkhand, Gujarat and Andhra Pradesh and for the city Chandigarh. Our analytical findings stated in the Theorem \ref{LASE0}, shows that the disease free equilibrium point $E^0$ is locally asymptotically stable with $\mathcal{R}_{0} <1$, and the Theorem \ref{trans-EEs} stated that a unique endemic equilibrium point $E^*$ is locally asymptotically stable for $\mathcal{R}_{0} > 1$. The numerical simulations of the SAIUQR model system (\ref{stateeq}) have been presented in the Figure \ref{F:LAS} for all the six individuals and for the different values of the disease transmission rate $\beta_s$. The value of the parameters considered for numerical simulations are $\alpha_a=0.264,$ $\alpha_i=0.76,$ $\alpha_u=0.96,$ $\gamma_a=0.0012,$ $\gamma_q=0.0015,$ $\delta=0.03,$ $\Lambda_s=1200$ and the other model parameter values are listed in the Table \ref{parval}. Six initial population sizes are considered for the model simulation, namely $(39402, 1500, 2000, 20, 0, 0)$, $(31402, 1200, 1500, 16, 0, 0)$, $(25402, 900, 1000, 12, 0, 0)$, $(20402, 600, 500, 8, 0, 0)$, $(15402, 300, 100, 4, 0, 0)$ and $(15000, 100, 50, 1, 0, 0)$. The time series simulation have been displayed for $\beta_s=1.10$ (red curves in the Figure \ref{F:LAS}) and $\beta_s=0.55$ (blue curves in the Figure \ref{F:LAS}). The values of $R_0$ are $1.2889$ and $0.7030$ for $\beta_s=1.10$ and $\beta_s=0.55$, respectively.  The blue curves in the Figure \ref{F:LAS} shows that the disease free equilibrium point $E^0(40000, 0, 0, 0, 0, 0)$ is locally asymptotically stable as well as globally asymptotically stable with $R_0=0.7030 < 1$. Our SAIUQR model system (\ref{stateeq}) converges to the endemic equilibrium point $E^*(31035.0, 1146.0, 17.6, 1.6, 2601.5, 5198.4)$ for $\beta_s=1.10$ and $R_0=1.2889 > 1$ (red curves), which has been displayed in the Figure \ref{F:LAS}. Hence, this numerical simulation verifies the analytical findings in Theorem \ref{LASE0} and Theorem \ref{trans-EEs}.

The Theorem \ref{trans-EEs} states that the SAIUQR model system (\ref{stateeq}) undergoes a transcritical bifurcation at the threshold $\mathcal{R}_{0} = 1$. We have plotted the COVID-19 reported symptomatic individuals $(I)$ in the $(R_0, I)$ plane by gradually increasing the disease transmission rate $\beta_s$ (see the Figure \ref{F:BF}). The model parameter values are  $\alpha_a=0.264,$ $\alpha_i=0.76,$ $\alpha_u=0.96,$ $\gamma_a=0.0012,$ $\gamma_q=0.0015,$ $\delta=0.03,$ $\Lambda_s=1200$ and the other parameter values listed in the Table \ref{parval}. We vary the disease transmission rate $\beta_s$ from 0.67 to 1.10 and computed the basic reproduction number $\mathcal{R}_{0}$ and the COVID-19 reported symptomatic individuals $(I)$. The numerically computed values are presented in the Figure \ref{F:BF}, which clearly shows that the system (\ref{stateeq}) experiences transcritical bifurcation at the threshold $\mathcal{R}_{0} = 1$. The blue curve represents the stable endemic equilibrium point $E^*$, black line represents the stable disease free equilibrium point $E^0$ and the red line represents the unstable branch of the disease free equilibrium point $E^0$. Hence, the Figure \ref{F:BF} shows that disease free equilibrium point $E^0$ is stable for the reproduction number $\mathcal{R}_{0} < 1$ and an endemic equilibrium point $E^*$ is stable for the reproduction number $\mathcal{R}_{0} > 1$. From the biological point of view, it can be described that the model system (\ref{stateeq}) will be free from COVID-19 for the reproduction number $\mathcal{R}_{0} < 1$ and the coronavirus diseases will spread throughout the people for $\mathcal{R}_{0} > 1.$

Figure \ref{F:R0}(a) represents the reproduction number $\mathcal{R}_{0}$ decreases as the recovery rate $\eta_i$ of reported infected individuals increases and the reproduction number $\mathcal{R}_{0}$ becomes less than one for $\beta_s=0.85$ and $\beta_s=0.76$. This indicates that the disease free equilibrium point $E^0$ switches the stability of the model system (\ref{stateeq}) as $\eta_i$ changes. But the reproduction number $\mathcal{R}_{0}$ remains grater than one for the disease transmission rate $\beta_s=1.10$ and $\beta_s=0.95$, that is, for large $\beta_s$ unique endemic equilibrium point remains locally asymptotically stable even if $\eta_i$ changes. In terms of COVID-19 diseases this interprets that if the rate of recovery for infected individuals $(\eta_i)$ be increased, which can be done by vaccinees or specific therapeutics, the model system (\ref{stateeq}) changes its stability to disease free equilibrium $E^0$ from endemic equilibrium $E^*$ but if the transmission of the disease $(\beta_s)$ is high enough then by vaccines or specific therapeutics, the system (\ref{stateeq}) can not change its stability from endemic equilibrium to disease free equilibrium.

Figure \ref{F:R0}(b) represents the reproduction number $\mathcal{R}_{0}$ increases as $\gamma_a$ (transition rate from asymptomatic individuals to symptomatic individuals) increases but the reproduction number $\mathcal{R}_{0}$ remains less than one for the disease transmission rate $\beta_s=0.55$ and $\beta_s=0.45$. For $\beta_s=0.66$ and $\beta_s=0.76$, the basic reproduction number $\mathcal{R}_{0}$ becomes grater than one and the SAIUQR model system (\ref{stateeq}) loses the stability of disease free equilibrium point $E^0$. Thus, to flatten the COVID-19 curve in any of the four provinces of India, reduction of the transmission of the COVID-19 disease is utmost priority even if the recovery rate increased by medication. Biologically it means that to mitigate the COVID-19 diseases, the people must have to maintain the social distancing, contact tracing by avoiding the mass gathering.

The predictive competency for the SAIUQR model system (\ref{stateeq}) requires valid estimation of the system parameters $\gamma_a$ (rate of transition from asymptomatic to symptomatic infected individuals), $\gamma_q$ (the rate that the quarantine become susceptible), $\theta$ (fraction of asymptomatic infectious that become reported symptomatic infectious), $\xi_a$ (rate at which asymptomatic individuals become quarantined). In the Figure \ref{F:R0_Surf}(a), we plot the reproduction number $\mathcal{R}_{0}$ as a function of $\xi_a$ and $\gamma_q$ for the parameter values in the Table \ref{parval} and estimated parameters for the state Jharkhand, to encapsulate the significance of these values in the evolution of COVID-19 outbreak. From the Figure \ref{F:R0_Surf}(a), it can be observed that the parameters have a little influence on the outbreak of the coronavirus diseases as the parameters $\xi_a$ and $\gamma_q$ has a little contribution for the reproduction number $\mathcal{R}_{0}$. In the Figure \ref{F:R0_Surf}(b), we plot the reproduction number $\mathcal{R}_{0}$ as a function of $\theta$ and $\gamma_a$ for the parameter values in the Table \ref{parval} and estimated parameters for the state Jharkhand, to encapsulate the significance of these values in the evolution of COVID-19 outbreak. The Figure \ref{F:R0_Surf}(b) shows that the parameters $\theta$ and $\gamma_a$ are more influential in increasing the reproduction number $\mathcal{R}_{0}$. Thus, to control the outbreak of COVID-19, we must control the parameters $\theta$ and $\gamma_a$. The correctness of these values rely  on the input of medical and biological epidemiologists. Thus, the fraction $\theta$ of reported symptomatic infected individuals may be substantially increased by public health reporting measures, with greater efforts to recognize all the present cases. Our model simulation reveals the effect of an increase in this fraction $\theta$ in the value of the reproduction number $\mathcal{R}_{0}$, as evident in the Figure \ref{F:R0_Surf}(b), for the COVID-19 epidemic in the four sates of India.

\subsection{Short-term prediction}

Mathematical modeling of infectious diseases can provide short-term and long-term prediction of the pandemic \cite{Giordano20,KhCov20,Liu20,Sarkar20}. Due to absence of any licensed vaccines or specific therapeutics, forecasting is of utmost importance for strategies to control and prevention of the diseases with limited resources. It should be noted that here we can predict the epidemiological traits  of SARS-CoV-2 or COVID-19 for short-term only as the Governmental strategies can be altered time to time resulting in the corresponding changes in the associated parameters of the proposed SAIUQR model. Also, it is true that the scientists are working for drugs and/or effective vaccines against COVID-19 and the presence of such pharmaceutical interventions will substantially change the outcomes \cite{Corey20}. Thus, in this study we performed a short-term prediction for our SAIUQR model system (\ref{stateeq}) using the parameter values in the Table \ref{parval} and the estimated parameter values in the Table \ref{parvalEst}.  Using the observed data up to May 24, 2020, a short-term prediction (for 20 days) has been done for daily new COVID-19 cases (first column) and cumulative confirmed cases (second column) are presented in the Figure \ref{F:Pred}. The black dot-dashed curve represents the short-term prediction of our SAIUQR model from May 25, 2020 to June 13, 2020.  Red shaded region is the standard deviation band of our SAIUQR model simulated curve. The standard deviations are computed from the model simulation based on the estimated data. In each of four states, we plot the standard deviation bands at a standard deviation level above and below the model simulation for different days. Standard deviation band gives an estimation of the deviation of the actual model data. The trend of the  predicted daily COVID-19 cases is increasing for all the four provinces of India. Prediction of the SAIUQR model should be regarded as an estimation of the daily infected population and cumulative confirmed cases of the four states of India. From the SAIUQR model simulation, we can predict that the estimated daily new reported COVID-19 cases on June 13, 2020 will be approximately 15, 454, 12, and 96 in Jharkhand, Gujarat, Chandigarh, and Andhra Pradesh, respectively (see the left column of the Figure \ref{F:Pred}). Our SAIUQR model simulation predict that the confirmed cumulative number of cases on June 13, 2020 will be approximately 661, 23955, 514, and 4487 in Jharkhand, Gujarat, Chandigarh, and Andhra Pradesh, respectively (see the right column of the Figure \ref{F:Pred}).

\begin{table}[h]
\caption{The SAIUQR model parameter values estimated from the observed daily new COVID-19 cases for four provinces of India, namely Jharkhand, Gujarat, Chandigarh and Andhra Pradesh. Six important parameters $\beta_s$, $\alpha_a$, $\alpha_i$, $\alpha_u$, $\gamma_a$ and $\gamma_q$ are estimated among fourteen system parameters.}
\begin{center}
 \begin{tabular}{l l l l l l l}
\hline
 Provinces & $\beta_s$  & $\alpha_a$  & $\alpha_i$   & $\alpha_u$ & $\gamma_a$ & $\gamma_q$\\
 \hline
Jharkhand       & 0.760 & 0.264 & 0.760 & 0.9600 & 0.0012 & 0.0015\\
Gujarat         & 1.006 & 0.342 & 0.168 & 0.1308 & 0.0004 & 0.0046\\
Chandigarh      & 0.750 & 0.294 & 0.444 & 0.4600 & 0.0010 & 0.0011\\
Andhra Pradesh  & 0.431 & 0.419 & 0.688 & 0.7100 & 0.0006 & 0.0280\\
\hline
\end{tabular}
\end{center}
\label{parvalEst}
\end{table}

\begin{table}[h]
\caption{\emph{Initial population size and the values of $\Lambda_s$ used in numerical simulations for four different provinces of India, namely Jharkhand, Gujarat, Chandigarh, and Andhra Pradesh.}}
\begin{center}
 \begin{tabular}{l l l l l l l l}
\hline
 Provinces & $S(0)$  & $A(0)$  & $Q(0)$   & $I(0)$ & $U(0)$ & $R(0)$ & $\Lambda_s$\\
 \hline
Jharkhand  & 39402 & 575  & 19 & 1 & 0 & 0 & 1200\\
Gujarat    & 85402 & 1525 & 27 & 1 & 0 & 0 & 1300\\
Chandigarh & 20402 & 275  & 10 & 1 & 0 & 0 & 1200\\
Andhra Pradesh & 75401 & 355 & 12 & 1 & 0 & 0 & 970\\
\hline
\end{tabular}
\end{center}
\label{parvalIC}
\end{table}

\begin{figure*}[ht!]
\centering
\includegraphics[width=0.9\textwidth]{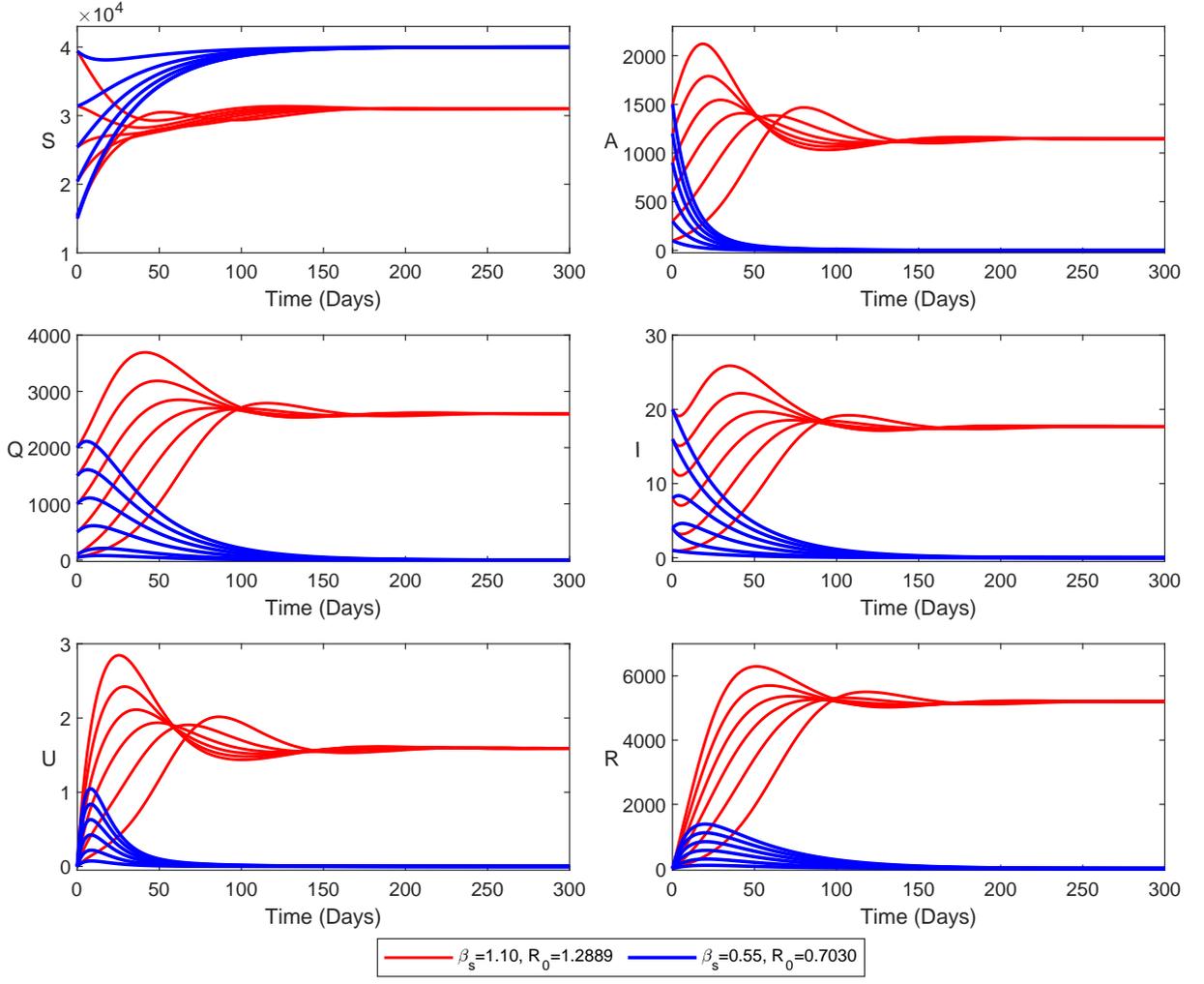}
\caption{Stability of the SAIUQR model system (\ref{stateeq}) around the disease free equilibrium point $E^0$ and an unique endemic equilibrium point $E^*$. The values of the estimated parameters are $\alpha_a=0.264,$ $\alpha_i=0.76,$ $\alpha_u=0.96,$ $\gamma_a=0.0012,$ $\gamma_q=0.0015,$ $\delta=0.03,$ $\Lambda_s=1200$ and other parameter values are listed in the Table \ref{parval}. Initial population sizes are $(39402, 1500, 2000, 20, 0, 0)$, $(31402, 1200, 1500, 16, 0, 0)$, $(25402, 900, 1000, 12, 0, 0)$, $(20402, 600, 500, 8, 0, 0)$, $(15402, 300, 100, 4, 0, 0)$, $(15000, 100, 50, 1, 0, 0)$. Time series solution for $\beta_s=1.10$ (red curves) and $\beta_s=0.55$ (blue curves). The values of $\mathcal{R}_0$ are $1.2889$ and $0.7030$ for $\beta_s=1.10$ and $\beta_s=0.55$, respectively. Disease free equilibrium point $E^0$ is locally asymptotically stable when $\mathcal{R}_0 < 1$ (blue curves) and the endemic equilibrium point $E^*$ is locally asymptotically stable when $\mathcal{R}_0 > 1$ (red curves).}
\label{F:LAS}
\end{figure*}


\begin{figure}
\includegraphics[width=0.45\textwidth]{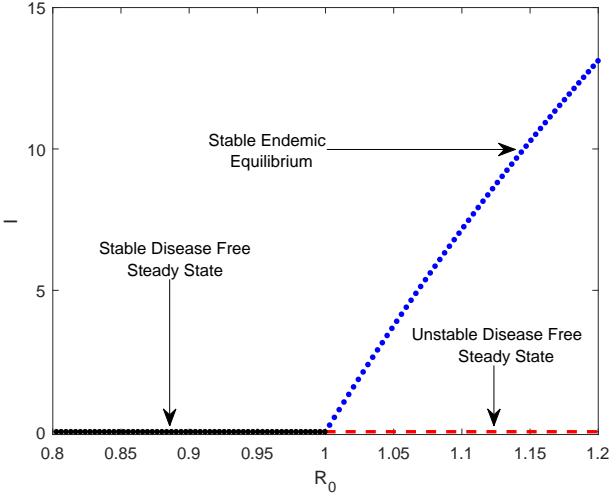}
\caption{ The figure represents the transcritical bifurcation diagram of the SAIUQR model system (\ref{stateeq}) with respect to the basic reproduction number $\mathcal{R}_0$. The parameter values are $\alpha_a=0.264,$ $\alpha_i=0.76,$ $\alpha_u=0.96,$ $\gamma_a=0.0012,$ $\gamma_q=0.0015,$ $\delta=0.03,$ $\Lambda_s=1200$ and other parameters as listed in the Table \ref{parval}. Stability of the SAIUQR system (\ref{stateeq}) exchange at the threshold $\mathcal{R}_0 = 1$.}
\label{F:BF}
\end{figure}

\begin{figure*}[ht]
\centering
\includegraphics[width=0.9\textwidth]{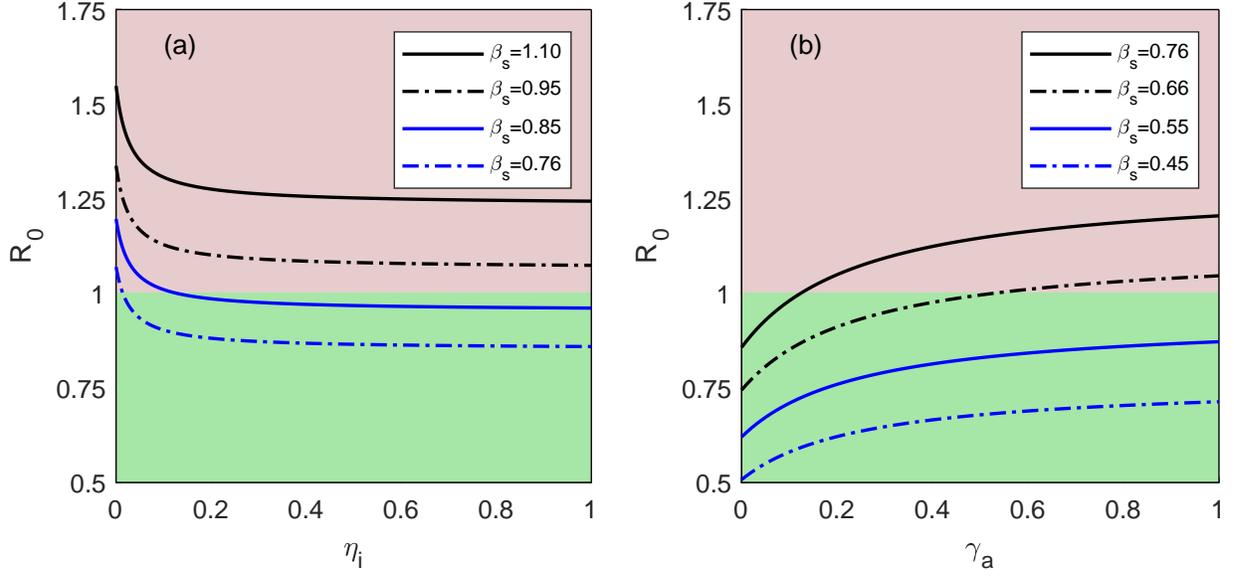}
\caption{The figures represents the basic reproduction number $\mathcal{R}_0$ in terms of (a) $\eta_i$ (rate of recovery for infected individuals) and (b) $\gamma_a$ (rate at which asymptomatic individuals develops detected symptomatic infected individuals). Green shaded region indicates $\mathcal{R}_0 < 1$ whereas the pink shaded region indicates $\mathcal{R}_0 > 1$. The parameter values are $\alpha_a=0.264,$ $\alpha_i=0.76,$ $\alpha_u=0.96,$ $\gamma_a=0.0012,$ $\gamma_q=0.0015,$ $\delta=0.03,$ $\Lambda_s=1200$ and the other parameter values are listed in the Table \ref{parval}. }
\label{F:R0}
\end{figure*}

\begin{figure*}[ht]
\centering
\includegraphics[width=0.9\textwidth]{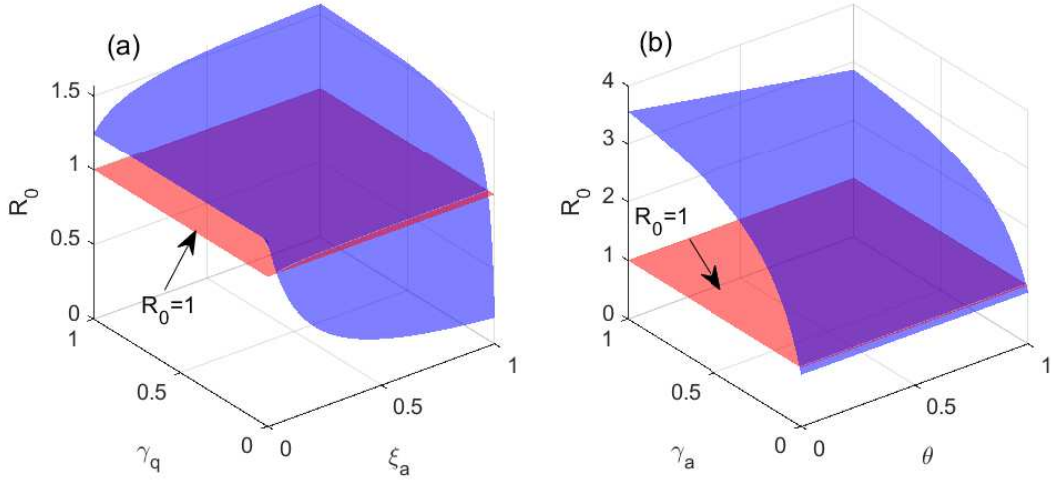}
\caption{The figures represents the surface plot of the basic reproduction number $\mathcal{R}_0$ in (a) ($\gamma_q$, ~$\xi_a$)-plane and (b) ($\gamma_a$, ~$\theta$)-plane. Red shading plane indicates $\mathcal{R}_0 = 1$. The parameter values for the sub-figure (a): $\beta_s=0.76$, $\alpha_a=0.264,$ $\alpha_i=0.76,$ $\alpha_u=0.96,$ $\gamma_a=0.0012,$ $\delta=0.03,$ $\Lambda_s=1200$ and for the sub-figure (b): $\beta_s=0.76$, $\alpha_a=0.264,$ $\alpha_i=0.76,$ $\alpha_u=0.96,$ $\gamma_q=0.0015,$ $\delta=0.03,$ $\Lambda_s=1200$ and the other parameter values are listed in the Table \ref{parval}. }
\label{F:R0_Surf}
\end{figure*}

\begin{figure*}[ht!]
\centering
\includegraphics[width=0.9\textwidth]{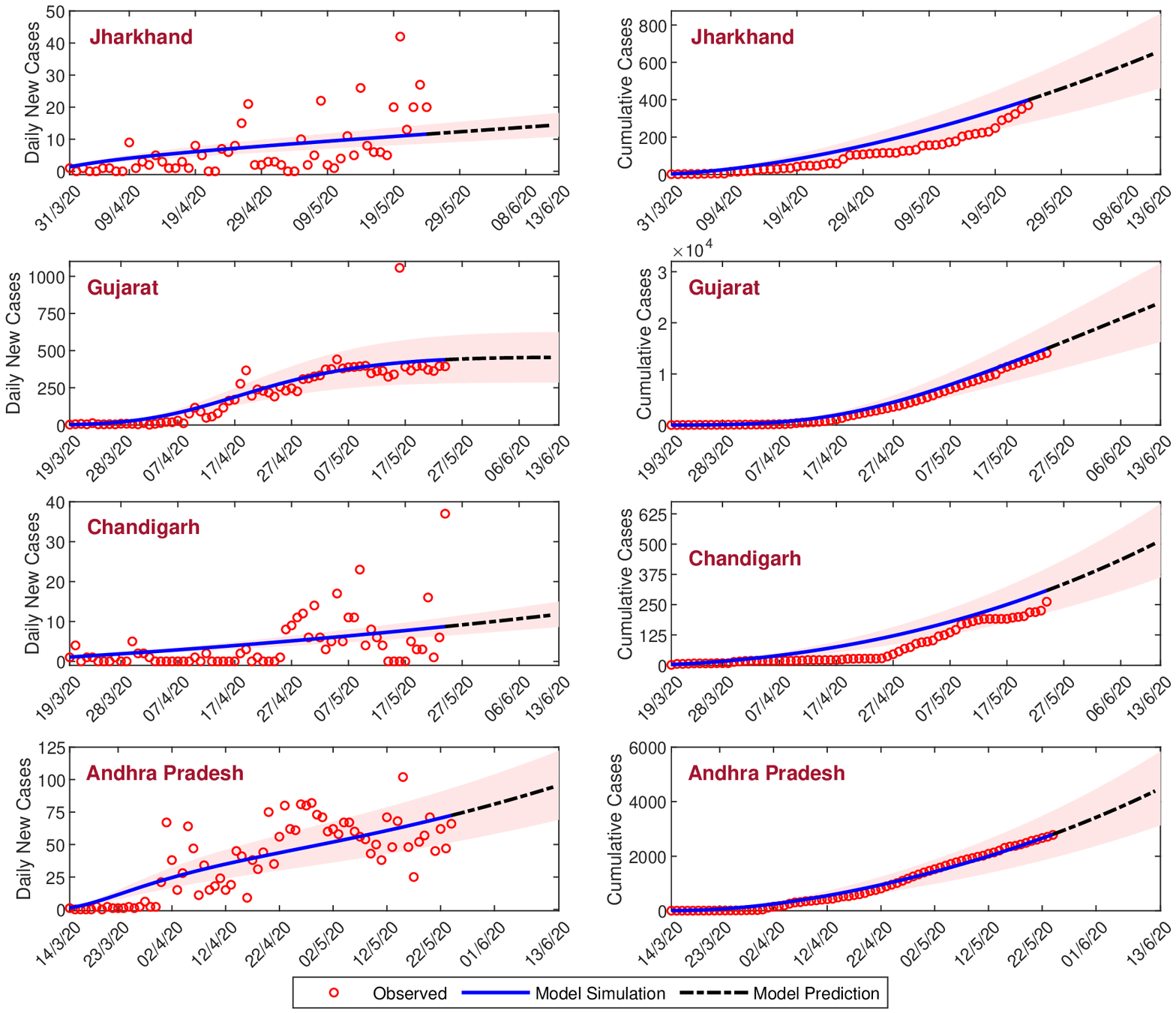}
\caption{The figure represents the short-term prediction of the daily new COVID-19 cases (first column) and the cumulative confirmed cases (second column) for the three states of India namely Jharkhand, Gujarat, Andhra Pradesh and one city of India namely Chandigarh. The black dot-dashed curve represents the prediction from May 25, 2020 to June 13, 2020 (20 days). The red shaded region is the standard deviation band of the SAIUQR model simulated curve.}
\label{F:Pred}
\end{figure*}

\section{Discussion and conclusion}

The SARS-CoV-2 pandemic in India is a potential menace throughout the country due to its exponential growth. Everyday the new cases are reported around 5-6 thousands and more than that from different states and territories  of India, which is an alarming situation as with the second most populated country worldwide \cite{covid-19tr}. Due to absence of any licensed vaccine, therapeutics or treatment and with a peculiar epidemiological traits of SARS-CoV-2, one would depend on the qualitative control of the diseases rather than complete elimination. During this period of an epidemic when person-to-person transmission is confirmed and the reported cases of SARS-CoV-2 viruses are rising throughout the globe, prediction is of utmost priority for health care planning and to manage the virus with limited resource. Furthermore, mathematical modeling can be a powerful tool in designing strategies to manage exponentially spreading coronavirus diseases in absence of any antivirals or diagnostic test.

In this study, we proposed and analyzed a new compartmental mathematical model for SARS-CoV-2 viruses to forecast and control the outbreak. In the model formulation, we incorporate the transmission variability of asymptomatic and unreported symptomatic individuals. We also incorporate the symptomatic infected populations who are reported by the public health services. We assume that reported infected individuals will no-longer associated into the infection as they are isolated and move to the hospitals or Intensive Care Unit (ICU). In our model, we incorporate the constant transmission rate in the early exponential growth phase of the SARS-CoV-2 diseases as identified in \cite{KhCov20,Gumel04}. We model the role of the Govt. imposed restrictions for the public in India, beginning on March 25, 2020, as a time-dependent decaying transmission rate after March 25, 2020. But, due to less stringent lockdown exponentially increasing the disease transmission rate, we are able to fit with increasing accuracy, our model simulations to the Indian reported cases data up to May 24, 2020.

We fit our SAIUQR model for the daily confirmed cases and cumulative confirmed cases of the four different provinces of India, namely Jharkhand, Andhra Pradesh, Chandigarh, and Gujarat with data up to May 24, 2020. The estimated model parameters for different states of India are given in the Table \ref{parvalEst} and the corresponding initial population size are listed in the Table \ref{parvalIC}. It can be observed that the basic reproduction number for four different provinces of India, namely Jharkhand, Andhra Pradesh, Chandigarh, and Gujarat are 1.6877, 1.2435, 1.4775, and 1.8830, respectively, which demonstrates the disease transmission rate is quite high that basically indicates the substantial outbreak of the COVID-19 diseases. This higher value of reproduction number $\mathcal{R}_{0}$ captures the outbreak of COVID-19 phenomena in India. Based on the estimated parameter values our model simulation suggest that the rate of disease transmission need to be controlled, otherwise India will enter in stage-3 of SARS-CoV-2 disease transmission within in a short period of time.

Based on the estimated model parameters we have validated our detailed analytical findings. Our proposed SAIUQR model has two biologically feasible singular points, namely infection free steady state $E^0$ and a unique endemic steady state $E^*$ and they become locally asymptotically stable for $\mathcal{R}_{0} < 1$ and $\mathcal{R}_{0} > 1$, respectively. Analytically, we have shown that the infection free steady state $E^0$ of the SAIUQR model (\ref{stateeq}) is globally asymptotically stable for $\mathcal{R}_{0} < 1.$ We also showed that the SAIUQR model (\ref{stateeq}) experiences transcritical bifurcation at the threshold parameter $\mathcal{R}_{0} = 1$, which has been shown in the Figure \ref{F:BF}. The Figure \ref{F:LAS} (blue curves) and the Figure \ref{F:LAS} (red curves) represents the local asymptotic stability as well as global asymptotic stability of the infection free steady state $E^0$ for $\mathcal{R}_{0} < 1$ and endemic steady state $E^*$ for $\mathcal{R}_{0} > 1$, respectively.

The calibrated model then utilized for short-term predictions in the four different states of India. Our SAIUQR model performs well in case of all the four provinces of India, namely Jharkahnd, Chandigarh, Gujarat,  and Andhra Pradesh for daily confirmed cases and cumulative confirmed cases. Albeit, the increasing (or exponential) pattern of daily new cases and cumulative confirmed cases of SARS-CoV-2 is well captured by our proposed model for all the four states of India, which has been shown in the Figure \ref{F:DataFit}. Our model simulation showed a short-term prediction for 20 days (from May 25, 2020 to June 13, 2020) for daily confirmed cases and cumulative confirmed cases of the four provinces of India.
The short-term prediction for the four provinces of India will demonstrates the increasing pattern of the daily and cumulative cases in the near future (see the Figure \ref{F:Pred}). From the simulation, our model predict that on June 13, 2020 the daily confirmed cases of COVID-19 of the four provinces of India, namely Jharkhand, Gujarat, Chandigarh, and Andhra Pradesh will be 15, 454, 12, and 96, respectively  (see the left column of the Figure \ref{F:Pred}). Similarly, from the simulation, our model predict that on June 13, 2020 the cumulative confirmed cases of COVID-19 of the four provinces of India, namely Jharkhand, Gujarat, Chandigarh, and Andhra Pradesh will be 661, 23955, 514, and 4487, respectively  (see the right column of the Figure \ref{F:Pred}).

It is worthy to mention that the scientists or clinicians are working for effective vaccine or therapeutics to eradicate and/or control the outbreak of SARS-CoV-2 diseases and the existence of such pharmaceutical interventions will substantially change the outcomes \cite{Corey20,Lurie20}. Thus, in this study, we are mainly focusing on short-term predictions for the  COVID-19 pandemic and subsequently there would be a very little chance to alter in corresponding parametric space. But the framework of our present compartmental model provides some significant insights into the dynamics and forecasting of the spread and control of COVID-19. Moreover, our model simulation suggest that the quarantine, reported and unreported symptomatic individuals as well as government intervention polices like media effect, lockdown and social distancing can play a key role in mitigating the transmission of COVID-19.

\section*{Acknowledgments}
This study is supported by Science and Engineering Research Board (SERB) (File No. ECR/ 2017/000234), Department of Science \& Technology, Government of India. The authors are thankful to the anonymous reviewers for their careful reading and constructive suggestions/comments which helped in better exposition of the manuscript.

\section*{Data availability}
All the data used in this work has been obtained from official sources \cite{covid-19tr}. All data supporting the findings of this study are in the paper and available from the corresponding author on request.

\section*{Author contributions statement}
Subhas Khajanchi and Kankan Sarkar  designed and performed the research as well as wrote the paper. \\

\section*{Competing interests}
The authors declare that they have no conflict of interest.

\appendix
\section{Proof of Theorem \ref{Postivity}}
\label{PositivityProof}
\textbf{Proof.} To prove the positivity of the system (\ref{stateeq}), we show that any solution initiating from the non-negative octant $\mathbf{R}_{+}^{6}$ remains positive for all $t > 0$. In order to do this, we have to prove that on each hyperplane bounding the non-negative octant the vector field points into $\mathbf{R}_{+}^{6}$. For our system  (\ref{stateeq}), we observe that
\begin{eqnarray*}
	\frac{dS}{dt}\bigg|_{S=0} &=&  \Lambda_{s} + \rho_s \gamma_q Q ~\geq~ 0,\\
	\frac{dA}{dt}\bigg|_{A=0} &=&  \beta_{s} S \bigg( \alpha_i \frac{I}{N} + \alpha_u \frac{U}{N} \bigg) ~\geq~ 0,\\
	\frac{dI}{dt}\bigg|_{I=0} &=&  \theta \gamma_{a} A + (1-\rho_s)\gamma_q Q ~\geq~ 0,\\
	\frac{dU}{dt}\bigg|_{U=0} &=&  (1 - \theta ) \gamma_a A~\geq~ 0,   \\
	\frac{dQ}{dt}\bigg|_{Q=0} &=& \xi_{a} A ~\geq~ 0,\\
	\frac{dR}{dt}\bigg|_{R=0} &=&  \eta_u U + \eta_i I + \eta_a A ~\geq~ 0.
\end{eqnarray*}

Therefore, the positivity of the solutions starting in the interior of $\mathbf{R}_{+}^{6}$ is assured. $\mathbf{R}_{+}^{6}$ is positively invariant set of the SAIUQR model system (\ref{stateeq}).	

\section{Proof of Theorem \ref{Bounded}}
\label{BoundedProof}
\textbf{Proof.} To prove the boundedness of the SAIUQR system (\ref{stateeq}), we add all the model equations, which gives $N = S + A + I + U + Q + R$. Taking the differentiation gives
\begin{eqnarray*}
	\frac{dN}{dt} &=& \Lambda_{s}- \delta N,
\end{eqnarray*}
which gives
\begin{eqnarray*}
	\limsup_{t \rightarrow \infty} N(t) & \leq & \frac{\Lambda_{s}}{\delta}.
\end{eqnarray*}
Without any loss of generality, we can assume that $\limsup_{t \rightarrow \infty} S(t) \leq \frac{\Lambda_{s}}{\delta},$ $\limsup_{t \rightarrow \infty} A(t) \leq \frac{\Lambda_{s}}{\delta},$ $\limsup_{t \rightarrow \infty} I(t) \leq \frac{\Lambda_{s}}{\delta},$ $\limsup_{t \rightarrow \infty} U(t) \leq \frac{\Lambda_{s}}{\delta},$ $\limsup_{t \rightarrow \infty} Q(t) \leq \frac{\Lambda_{s}}{\delta}$, and $\limsup_{t \rightarrow \infty} R(t) \leq \frac{\Lambda_{s}}{\delta}.$ Thus, we have a bounded set
\begin{eqnarray*}
	\Omega & =& \bigg\{ (S, A, I, U, Q, R) \in \mathbf{R}_{+}^{6} : 0 \leq S, A, I, U, Q, R \leq \frac{\Lambda_{s}}{\delta} \bigg\},
\end{eqnarray*}

\noindent which is also a positively invariant set with respect to the SAIUQR model (\ref{stateeq}). Thus, any solution trajectory starting from an interior point of $\mathbf{R}_{+}^{6}$ ultimately enter into the region $\Omega$ and remains there for all finite time. This results indicates that none of the individuals grow unboundedly or exponentially for a finite time window.

\section{Proof of Theorem \ref{LASE0}}
\label{LASE0Proof}
\textbf{Proof.}	The variational matrix around the infection free steady state $E^{0}$ for the SAIUQR model system (\ref{stateeq}) is given by
	\begin{widetext}
		\begin{eqnarray*}
			J_{E^0} &=& \left( \begin{array}{cccccc}
				-\delta & -\beta_s\alpha_a & -\beta_s\alpha_i & -\beta_s\alpha_u & \rho_s\gamma_q & 0 \\
				0 & \beta_s\alpha_a - (\xi_a+\gamma_a+\eta_a+\delta) & \beta_s\alpha_i & \beta_s\alpha_u & 0 & 0 \\
				0 & \theta\gamma_a & -(\eta_i+\delta) & 0 & (1-\rho_s)\gamma_q & 0 \\
				0 & (1-\theta)\gamma_a & 0 & -(\eta_u+\delta) & 0 & 0 \\
				0 & \xi_a & 0 & 0 & -(\gamma_q+\delta) & 0 \\
				0 & \eta_a & \eta_i & \eta_u & 0 & -\delta
			\end{array}\right).
		\end{eqnarray*}
	\end{widetext}
The above Jacobian matrix $J_{E^0}$ has two repeated eigenvalues which are -$\delta$, while the other four eigenvalues are the roots of the following characteristics equation $|J_{E^0} - \lambda I| = 0$,
	\begin{eqnarray*}
		(\xi_a &+& \gamma_a + \eta_a + \delta + \lambda)(\eta_i + \delta + \lambda)(\eta_u + \delta + \lambda)(\gamma_q + \delta + \lambda)\\
		&-& \beta_s \alpha_a (\eta_i + \delta + \lambda)(\eta_u + \delta + \lambda)(\gamma_q + \delta + \lambda)\\
		&-& \theta \gamma_a \beta_s \alpha_i  (\eta_u + \delta + \lambda)(\gamma_q + \delta + \lambda)\\
		&-& (1-\theta)\gamma_a \beta_s\alpha_u (\eta_i + \delta + \lambda)(\gamma_q + \delta + \lambda)\\
		&-& (1-\rho_s)\gamma_q \xi_a \beta_s \alpha_i (\eta_u + \delta + \lambda) ~=~0,
	\end{eqnarray*}
	which can be rewritten in the following form:
	\begin{eqnarray*}
		&~&\frac{\beta_s \alpha_a}{\xi_a + \gamma_a + \eta_a + \delta + \lambda}  +  \frac{\theta \gamma_a \beta_s \alpha_i}{(\xi_a + \gamma_a + \eta_a + \delta + \lambda)(\eta_i + \delta + \lambda)}\\
		&+& \frac{(1-\theta)\gamma_a \beta_s\alpha_u}{(\xi_a + \gamma_a + \eta_a + \delta + \lambda)(\eta_u + \delta + \lambda)}\\
		&+& \frac{(1-\rho_s)\gamma_q \xi_a \beta_s \alpha_i}{(\xi_a + \gamma_a + \eta_a + \delta + \lambda)(\eta_i + \delta + \lambda)(\gamma_q + \delta + \lambda)} ~=~ 1.
	\end{eqnarray*}
	Denote the above expression as the following:
	\begin{eqnarray*}
		m_1 (\lambda) &=& \frac{\beta_s \alpha_a}{\xi_a + \gamma_a + \eta_a + \delta + \lambda}\\
		&+&  \frac{\theta \gamma_a \beta_s \alpha_i}{(\xi_a + \gamma_a + \eta_a + \delta + \lambda)(\eta_i + \delta + \lambda)}\\
		&+& \frac{(1-\theta)\gamma_a \beta_s\alpha_u}{(\xi_a + \gamma_a + \eta_a + \delta + \lambda)(\eta_u + \delta + \lambda)}  \\
		&+& \frac{(1-\rho_s)\gamma_q \xi_a \beta_s \alpha_i}{(\xi_a + \gamma_a + \eta_a + \delta + \lambda)(\eta_i + \delta + \lambda)(\gamma_q + \delta + \lambda)},  \\[0.2cm]
		&=& m_{11} (\lambda) + m_{22}(\lambda) + m_{33}(\lambda) + m_{44}(\lambda)~~~\textrm{(say)}.
	\end{eqnarray*}
	Substitute $\lambda = x + i y$, and we know that $Re(\lambda) \geq 0$, then the above expression leads to
	\begin{eqnarray*}
		| m_{11} (\lambda) | &\leq & \frac{\beta_s \alpha_a}{| \xi_a + \gamma_a + \eta_a + \delta + \lambda |} ~\leq~ m_{11} (x)\\
		~&\leq&~ m_{11} (0),   \\
		| m_{22} (\lambda) | &\leq & \frac{\theta \gamma_a \beta_s \alpha_i}{| \xi_a + \gamma_a + \eta_a + \delta + \lambda | |\eta_i + \delta + \lambda|} ~\leq~ m_{22} (x)\\
		~&\leq&~ m_{22} (0),   \\
		| m_{33} (\lambda) | &\leq & \frac{(1-\theta)\gamma_a \beta_s\alpha_u}{|\xi_a + \gamma_a + \eta_a + \delta + \lambda| |\eta_u + \delta + \lambda|}\\
		~&\leq&~ m_{33} (x) ~\leq~ m_{33} (0),   \\
		| m_{44} (\lambda) | &\leq & \frac{(1-\rho_s)\gamma_q \xi_a \beta_s \alpha_i}{|\xi_a + \gamma_a + \eta_a + \delta + \lambda| |\eta_i + \delta + \lambda| |\gamma_q + \delta + \lambda|}\\
		~&\leq&~ m_{44} (x) ~\leq~ m_{44} (0).
	\end{eqnarray*}
	Thus, $m_{11}(0) + m_{22}(0) +  m_{33}(0) +  m_{44}(0) ~=~ m_{1}(0) ~=~ \mathcal{R}_0 < 1,$ which gives that $m_{1} (\lambda) \leq 1.$ Therefore for $\mathcal{R}_0 < 1$, all the eigenvalues of the characteristics equation $m_1 (\lambda) = 1$ are real or have negative real parts. Thus for $\mathcal{R}_0 < 1$, all the eigenvalues are negative and hence the infection free steady state $E^0$ is locally asymptotically stable.    \\
	
	\noindent Now, if we consider $\mathcal{R}_0 > 1$, that is, $m_1 (0) > 1$, then
	\begin{eqnarray*}
		\lim_{\lambda \rightarrow \infty}  m_{1} (\lambda) &=& 0,
	\end{eqnarray*}
	then there exists $\lambda_{1}^{*} > 0$ in such a way that $ m_{1} (\lambda_{1}^{*}) = 1.$ This indicates that there exists non-negative eigenvalue $\lambda_{1}^{*} > 0$ for the variational matrix $J_{E^0}$. Hence, the infection free steady state $E^0$  is unstable for  $\mathcal{R}_0 > 1$.

\section{Proof of Theorem \ref{GASE0}}
\label{GASE0Proof}
\textbf{Proof.}	To prove the global asymptotic stability of the infection free steady state $E^{0}(S^{0}, A^{0}, Q^{0}, I^{0}, U^{0}, R^{0})$, we can rewrite the SAIUQR model (\ref{stateeq}) in the following form:
	\begin{eqnarray*}
		\frac{dX}{dt} &=& F(X, V),    \\
		\frac{dV}{dt} &=& G(X, V),~~~G(X, 0) = 0,
	\end{eqnarray*}
	where $X = (S, R) \in \mathbf{R}^{2}$ represents (its components) the number of susceptible or uninfected individuals and $V = (A, I, U, Q) \in \mathbf{R}^{4}$ represents (its components) the number of infected individuals incorporating asymptomatic, quarantine, and infectious etc. $E^{0} = (X^*, 0)$ designates the infection free steady state for the SAIUQR model system (\ref{stateeq}). For the compartmental model (\ref{stateeq}), $F(X, V)$ and $G(X, V)$ are defined as follows:
	
	\begin{eqnarray*}
		F(X, V) &=& \left(
		\begin{array}{c}
			\Lambda_{s} - \beta_s \frac{S}{N} ( \alpha_a A + \alpha_i I + \alpha_u U) + \rho_{s} \gamma_q Q - \delta S \\
			\eta_u U + \eta_{i} I + \eta_{a} A - \delta R \\
		\end{array}
		\right),\\  ~~\textrm{and}~~\\
		G(X, V) &=& \left(
		\begin{array}{c}
			\beta_s \frac{S}{N} ( \alpha_a A + \alpha_i I + \alpha_u U) - (\xi_a + \gamma_a) A - \eta_a A - \delta A \\
			\theta \gamma_a A + (1-\rho_s) \gamma_q Q - \eta_i I - \delta I \\
			(1 - \theta) \gamma_{a} A - \eta_u U - \delta U \\
			\xi_a A - \gamma_q Q - \delta Q  \\
		\end{array}
		\right).
	\end{eqnarray*}
	From the above expression of $G(X, V)$, it is clear that $G(X, 0) = 0.$
	
	The following two conditions $(C1)$ and $(C2)$ must be met to assure the global asymptotic stability:   \\
	~~~~~~~~(C1) For $\frac{dX}{dt} = F(X, 0),$ $X^*$ is globally asymptotically stable,   \\
	~~~~~~~~(C2) $G(X, V) = BV - \widehat{G}(X, V),$ $\widehat{G}(X, V) \geq 0$ for $(X, V) \in \Omega$,   \\
	where $B = D_{I}G(X^{*}, 0)$ is an M-matrix (the non-diagonal components are non-negative) and in the region $\Omega$, the SAIUQR model system (\ref{stateeq}) is biologically feasible. The compartmental model (\ref{stateeq}) stated in the condition $(C1)$ can be expressed as
	\begin{eqnarray*}
		\frac{d}{dt} \left(
		\begin{array}{c}
			S \\
			R \\
		\end{array}
		\right)~=~\left(
		\begin{array}{c}
			\Lambda_{s} - \delta S \\
			- \delta R \\
		\end{array}
		\right).
	\end{eqnarray*}
	Analytically solve the above system of equations, we obtain that $S(t) = \frac{\Lambda_{s}}{\delta} + \exp(-\delta t) (S(0) - \frac{\Lambda_{s}}{\delta}),$ and $R(t) = \exp(-\delta t) R(0).$ Considering $t \rightarrow \infty$, $S(t) = \frac{\Lambda_{s}}{\delta}$ and $R(t) \rightarrow 0.$ Thus, $X^*$ is globally asymptotically stable for $\frac{dX}{dt} = F(X, 0).$ Thus, the first condition $(C1)$ holds for the system (\ref{stateeq}).  \\
	
	Now the matrices $B$ and $\widehat{G}(X, V)$ for the SAIUQR model system (\ref{stateeq}) can be expressed as
	\begin{widetext}
		\begin{eqnarray*}
			B &=& \left(
			\begin{array}{cccc}
				-(\xi_a +\gamma_a + \eta_a + \delta) + \alpha_a \beta_s & \alpha_i\beta_s & \alpha_u\beta_s & 0 \\
				\theta\gamma_a & -(\eta_i+\delta) & 0 & (1-\rho_s)\gamma_q \\
				(1-\theta)\gamma_a & 0 & -(\eta_u +\delta) & 0 \\
				\xi_a & 0 & 0 & -(\gamma_q+\delta) \\
			\end{array}
			\right),    \\
			\widehat{G}(X, V) &=& \left(
			\begin{array}{c}
				\alpha_a\beta_s A \bigg(1-\frac{S}{N}\bigg) + \alpha_i\beta_s I \bigg(1-\frac{S}{N}\bigg) + \alpha_u\beta_s U \bigg(1-\frac{S}{N}\bigg)\\
				0 \\
				0 \\
				0 \\
			\end{array}
			\right).
		\end{eqnarray*}
	\end{widetext}
	
	It is clear that $B$ is a M-matrix as all its non-diagonal components are non-negative. Also, $\widehat{G}(X, V) \geq 0$ in the region $\Omega$ as $S(t) \leq N(t).$ Also, we showed that $X^*$ is a globally asymptotically stable steady state of the system $\frac{dX}{dt} = F(X,0)$. Therefore, the infection free steady state $E^0$ of the SAIURQ model (\ref{stateeq})  is globally asymptotically stable in the region $\Omega$ for $\mathcal{R}_0 < 1.$

\section{Proof of Theorem \ref{trans-EEs}}
\label{trans-EEsProof}
\textbf{Proof.}	Now, we use the theory of center manifold to investigate the local asymptotic stability of the interior equilibrium point $E^{*}(S^{*}, A^{*}, Q^{*}, I^{*}, U^{*}, R^{*})$ by considering the disease transmission rate $\beta_s$ as a bifurcation parameter, $\beta_s = \beta_s^{c}$ corresponding to $\mathcal{R}_0 = 1,$ is
	\begin{widetext}
		\begin{eqnarray*}
			\beta_s^{c} &=& \frac{(\eta_u+\delta)(\gamma_q+\delta)(\eta_i+\delta)(\xi_a+\gamma_a+\eta_a+\delta)}{\alpha_a (\eta_u+\delta)(\gamma_q+\delta)(\eta_i+\delta) + (1-\theta)\alpha_u\gamma_a (\gamma_q+\delta)(\eta_i+\delta) + \alpha_i \left\{ \theta\gamma_a(\gamma_q+\delta) +(1-\rho_s)\xi_a\gamma_a \right\} (\eta_u+\delta) }.
		\end{eqnarray*}
	\end{widetext}
	
	The variational matrix of the SAIUQR model (\ref{stateeq}) at $\beta_s = \beta_s^{c}$, denoted by $J_{E^0}$ has the right eigenvector associated to zero is eigenvalue given by $\omega = [ \omega_1, ~\omega_2, ~\omega_3, ~\omega_4, ~\omega_5, ~\omega_6 ]^{T}$, where
	\begin{eqnarray*}
		\omega_1 &=& \frac{\omega_2}{\delta} \bigg[ - \frac{(1-\theta)\gamma_a\beta_s\alpha_u}{\eta_u+\delta} + \frac{\rho_s\gamma_q\xi_a}{\eta_q+\delta} - (\xi_a+\gamma_a+\eta_a+\delta)\\
		&+& \frac{\beta_s\alpha_u(1-\theta)\gamma_a}{\eta_u+\delta}  \bigg],   \\
		\omega_2 &=& \omega_2 ~>~0,~~~ \omega_4 ~=~ \frac{(1-\theta)\gamma_a \omega_2}{\eta_u+\delta},~~~ \omega_5 ~=~ \frac{\xi_a \omega_2}{\eta_q+\delta},    \\
		\omega_3 &=& \frac{\omega_2}{\beta_s\alpha_i} \bigg[ (\xi_a+\gamma_a+\eta_a+\delta) - \beta_s\alpha_a - \frac{\beta_s\alpha_u(1-\theta)\gamma_a}{\eta_u+\delta}  \bigg],   \\
		\omega_6 &=& \frac{\omega_2}{\delta} \bigg[ \eta_a + \frac{\eta_u(1-\theta)\gamma_a}{\eta_u+\delta} + \frac{\eta_i ((\xi_a+\gamma_a+\eta_a+\delta)-\beta_s\alpha_a)}{\beta_s\alpha_i}\\
		&-& \frac{\eta_i\beta_s\alpha_u(1-\theta)\gamma_a}{\beta_s\alpha_i(\eta_u+\delta)} \bigg].
	\end{eqnarray*}
	
	Similarly, at the threshold $\beta_s = \beta_s^{c}$, the variational matrix $J_{E^0}$ has the left eigenvector associated to zero eigenvalue is given by $\upsilon = [ \upsilon_1, ~\upsilon_2, ~\upsilon_3, ~\upsilon_4, ~\upsilon_5, ~\upsilon_6 ]$, where
	\begin{eqnarray*}
		\upsilon_1 &=& 0,~~~\upsilon_6 = 0,~~~ \upsilon_2 = \upsilon_2 ~>~0,\\
		~~~\upsilon_3 &=& \frac{\beta_s\alpha_i \upsilon_2}{\eta_i+\delta},~~~\upsilon_4 = \frac{\beta_s\alpha_u \upsilon_2}{\eta_u+\delta}, \\[0.2cm]
		\upsilon_5 &=& \frac{\upsilon_2}{\xi_a} \bigg[ (\xi_a+\gamma_a+\eta_a+\delta) - \beta_s\alpha_a - \frac{\theta\gamma_a\beta_s\alpha_i}{\eta_i+\delta}\\
		&-& \frac{\beta_s\alpha_u(1-\theta)\gamma_a}{\eta_u+\delta}  \bigg].
	\end{eqnarray*}

	Let us introduce the notations for the SAIUQR model system (\ref{stateeq}): $S = x_1$, $A = x_2$, $I = x_3$, $U = x_4$, $Q = x_5$, $R = x_6$,  and $\frac{dx_i}{dt} = f_i $, where $i = 1, 2,..., 6.$  Now, we compute the following nonzero second order partial derivatives of $f_i $ at the infection free steady state $E^{0}$ and obtain
	\begin{eqnarray*}
		\frac{\partial^{2}f_{2}}{\partial x_2 \partial x_3} &=& - \beta_s (\alpha_a + \alpha_i) \frac{\Lambda_s}{\delta},~ \frac{\partial^{2}f_{2}}{\partial x_2 \partial x_4} = - \beta_s (\alpha_a + \alpha_u) \frac{\Lambda_s}{\delta},\\
		\frac{\partial^{2}f_{2}}{\partial x_2 \partial x_5} &=& - \beta_s \alpha_a \frac{\Lambda_s}{\delta}, ~~~\frac{\partial^{2}f_{2}}{\partial x_2 \partial x_6} ~=~ - \beta_s \alpha_a \frac{\Lambda_s}{\delta},   \\
		\frac{\partial^{2}f_{2}}{\partial x_3 \partial x_4} &=& - \beta_s (\alpha_i + \alpha_u) \frac{\Lambda_s}{\delta}, ~~~\frac{\partial^{2}f_{2}}{\partial x_3 \partial x_5} ~=~ - \beta_s \alpha_i \frac{\Lambda_s}{\delta},\\
		\frac{\partial^{2}f_{2}}{\partial x_3 \partial x_6} &=& - \beta_s \alpha_i \frac{\Lambda_s}{\delta}, ~~~\frac{\partial^{2}f_{2}}{\partial x_4 \partial x_6} ~=~ - \beta_s \alpha_u \frac{\Lambda_s}{\delta},    \\
		\frac{\partial^{2}f_{2}}{\partial x_4 \partial x_5} &=& - \beta_s \alpha_u \frac{\Lambda_s}{\delta}, ~~\frac{\partial^{2}f_{2}}{\partial x_2 \partial x_2} ~=~ - 2 \beta_s \alpha_a \frac{\Lambda_s}{\delta},\\
		\frac{\partial^{2}f_{2}}{\partial x_3 \partial x_3} &=& - 2 \beta_s \alpha_i \frac{\Lambda_s}{\delta}, ~~~\frac{\partial^{2}f_{2}}{\partial x_4 \partial x_4} ~=~ - 2 \beta_s \alpha_u \frac{\Lambda_s}{\delta}.
	\end{eqnarray*}
	
	The rest of the partial derivatives at the infection free steady state $E^{0}$ remains zero. Now, we compute the coefficients $a$ and $b$ due to the well-known Theorem 4.1 by Castillo-Chavez \& Song \cite{Castillo04} as follows:
	\begin{eqnarray*}
		a &=& \sum_{i, j, k = 1}^{6} \upsilon_{k} \omega_i \omega_j  \frac{\partial^{2}f_{k} (0, 0)}{\partial x_i \partial x_j},
	\end{eqnarray*}
	and
	\begin{eqnarray*}
		b &=& \sum_{i, k = 1}^{6} \upsilon_{k} \omega_i \frac{\partial^{2}f_{k} (0, 0)}{\partial x_i \partial \beta_s}.
	\end{eqnarray*}
	By substituting the values of all the nonzero second-order partial derivatives and the left and right eigenvectors from above analysis at the threshold $\beta_s = \beta_s^{c}$, we have
	\begin{eqnarray*}
		a &=& - \frac{\beta_s \upsilon_{2} \Lambda_s}{\delta} ( \omega_2 \omega_3 (\alpha_a + \alpha_i) + \omega_2 \omega_4 (\alpha_a + \alpha_u)\\
		&+& \omega_2 \omega_5 \alpha_a + \omega_2 \omega_6 \alpha_a + \omega_3 \omega_4 (\alpha_i + \alpha_u) + \omega_3 \omega_5 \alpha_i\\
		&+& \omega_3 \omega_6 \alpha_i + \omega_4 \omega_6 \alpha_u + \omega_4 \omega_5 \alpha_u + 2 \omega_2^{2} \alpha_a\\
		&+& 2 \omega_3^{2} \alpha_i + 2 \omega_4^{2} \alpha_u ),
	\end{eqnarray*}
	and
	\begin{eqnarray*}
		b &=&  \frac{\upsilon_{2} \Lambda_s}{\delta} (\omega_2 \alpha_a + \omega_3 \alpha_i + \omega_4 \alpha_u ).
	\end{eqnarray*}
	
	From the above expressions, it can be observed that $a < 0$ and $b > 0$, therefore due to the Remark 1 of the well-known Theorem 4.1 by Castillo-Chavez \& Song \cite{Castillo04}, and by Khajanchi et al. \cite{Khajanchi18} a transcritical bifurcation occurs at the basic reproduction number $\mathcal{R}_0 = 1$ and the interior equilibrium point $E^*$ is locally asymptotically stable for $\mathcal{R}_0 > 1$.
	



\appendix~


\section*{References}
	
\begin{thebibliography}{100}

\bibitem{Wu20} Z. Wu, J.M. McGoogan,~``Characteristics of and important lessons from the coronavirus disease 2019 (COVID-19) outbreak in China: summary of a report of 72,314 cases from the Chinese center for disease control and prevention,''~JAMA~\textbf{323},~1239--1242~(2020).

\bibitem{Chen20} Y. Chen, Q. Liu, D. Guo,~``Emerging coronaviruses: Genome structure, replication, and pathogenesis".~J. Med. Virol.~1--6~(2020).

\bibitem{Huang20} C. Huang, Y. Wang, X. Li, L. Ren, J. Zhao, Y. Hu, L. Zhang, G. Fan, J. Xu, X. Gu, et al.~``Clinical features of patients infected with 2019 novel coronavirus in wuhan, china,"~The Lancet~\textbf{395}(10223),~497--506~(2020).

\bibitem{bbc1} BBC News,~https://www.bbc.com/news/world-52114829,~(Retrived on April 01, 2020).

\bibitem{WHO129} WHO.~Coronavirus Disease 2019 (COVID-19):~Situation Report 129~(WHO,~2020)~(Retrived on May 29, 2020).


\bibitem{Lina20}  Q. Lina,1, S. Zhaob, D. Gaod Y. Lou, S. Yang, S.S. Musa, M.H. Wang, Y. Cai, W. Wang, L. Yang, D. He,~``A conceptual model for the coronavirus disease 2019 (COVID-19) outbreak in Wuhan, China with individual reaction and governmental action,"~Int. J. Infect. Dis.~\textbf{93},~211--216~(2020).


\bibitem{IndCov} India Today.~https://www.indiatoday.in/india/story/kerala-reports-first-confirmed-novel-coronavirus-case-in-india-1641593-2020-01-30. ~(Retrived on January 30 2020).

\bibitem{Pulla20} P. Pulla,~``Covid-19: India imposes lockdown for 21 days and cases rise,"~BMJ~\textbf{368},~(2020).~https://doi.org/10.1136/bmj.m1251

\bibitem{mohfw} MOHFW,~Coronavirus disease 2019 (COVID-19).~https://www.mohfw.gov.in/,~(Retrieved on March 25, 2020)~(2020).

\bibitem{Kucharski20} A.J. Kucharski, T.W. Russell, C. Diamond, Y. Liu, J. Edmunds, S. Funk, R.M. Eggo,~``Early dynamics of transmission and control of COVID-19: a mathematical modelling study,"~Lancet Infect. Dis.~\textbf{20},~553--558~(2020).

\bibitem{Tang20} B. Tang, X. Wang, Q. Li, N.L. Bragazzi, S. Tang, Y. Xiao, J. Wu,~``Estimation of the transmission risk of the 2019-ncov and its implication for public health interventions,"~J. Clin. Med.~\textbf{9}(2),~462~(2020).

\bibitem{WuJT20} J.T. Wu, K. Leung, G.M. Leung,~``Nowcasting and forecasting the potential domestic and international spread of the 2019-nCoV outbreak originating in Wuhan, China: a modelling study,"~Lancet~\textbf{395},~689--697~(2020).

\bibitem{Fanelli20} D. Fanelli, F. Piazza,~``Analysis and forecast of COVID-19 spreading in China, Italy and France,"~Chaos Soliton Fract.~\textbf{134},~109761~(2020).

\bibitem{Ribeiro20} MHDM Ribeiro, RG Silva, VC Mariani, LS Coelho,~``Short term forecasting COVID-19 cumu-lative confirmed cases: Perspectives for Brazil,"~Chaos Soliton Fract.~\textbf{135}~109853~(2020).

\bibitem{Chakraborty20} T. Chakraborty, I. Ghosh~``Real-time forecasts and risk assessmentof novel coronavirus (COVID-19) cases: A data-driven analysis,"~Chaos Soliton Fract.~\textbf{135}~109850~ (2020).

\bibitem{Sarkar20} K. Sarkar, S. Khajanchi,~``Modeling and forecasting of the COVID-19 pandemic in India,"~arXiv:2005.07071~(2020).

\bibitem{He20} X. He, E.H. Lau, P. Wu, X. Deng, J. Wang, X. Hao, Y.C. Lau, et al.~``Temporal dynamics in viral shedding and transmissibility of COVID-19,"~Nat. Med.~\textbf{26},~672--675~(2020).

\bibitem{KhCov20} S. Khajanchi, K. Sarkar, J. Mondal, M. Perc,~``Dynamics of the COVID-19 pandemic in India,"~arXiv:2005.06286~(2020).

\bibitem{Anastassopoulou} C. Anastassopoulou, L. Russo, A. Tsakris, C. Siettos,~``Data-based analysis, modelling and forecasting of the COVID-19 outbreak,"~PLoS One~\textbf{15},~e0230405~(2020).

\bibitem{Giordano20} G. Giordano, F. Blanchini, R. Bruno, P. Colaneri, A. Di Filippo, A. Di Matteo, M. Colaneri,~``Modelling the COVID-19 epidemic and implementation of population-wide interventions in Italy,"~Nat. Med.~(2020).~https://doi.org/10.1038/s41591-020-0883-7.

\bibitem{Xiao-Lin20} J. Xiao-Lin, Z. Xiao-Li, Z. Xiang-Na et al.,~``Transmission Potential of Asymptomatic and Paucisymptomatic Severe Acute Respiratory Syndrome Coronavirus 2 Infections: A 3-Family Cluster Study in China,"~J. Infect. Dis.~\textbf{221},~1948--1952~(2020).

\bibitem{Ndairou20} F. Ndairou et al.,~``Mathematical modeling of COVID-19 transmission dynamics with a case study of Wuhan,"~Chaos Soliton Fract.~\textbf{135},~109846~(2020).

\bibitem{Mena20} R. H. Mena, J. X. Velasco-Hernandez, N. B. Mantilla-Beniers et al.,~``Using the posterior predictive distribution to analyse epidemic models: COVID-19 in Mexico city,"~arXiv:2005.02294~(2020).

\bibitem{Wong20} G. N. Wong et al.,~``Modeling COVID-19 dynamics in Illinois under non-pharmaceutical interventions,"~arXiv:2006.02036~(2020).

\bibitem{Dehning20} J. Dehning, J. Zierenberg, F.P. Spitzner, M. Wibral, J.P. Neto, M. Wilczek, V. Priesemann,~``Inferring change points in the spread of COVID-19 reveals the effectiveness of interventions,"~Science,~(2020)~eabb9789~doi:10.1126/science.abb9789.

\bibitem{Anderson91} R.M. Anderson, R.M. May,~``Infectious diseases of humans,"~(London,~Oxford University Press, 1991).

\bibitem{Diekmann00} O. Diekmann, J.A.P. Heesterbeek,~``Mathematical epidemiology of infectious diseases: model building, analysis and interpretation,"~(New York,~Wiley,~2000).

\bibitem{Hethcote00} H.W. Hethcote,~``The mathematics of infectious diseases,"~SIAM Rev.~\textbf{42},~599--653~(2000).

\bibitem{Gumel04} A.B. Gumel, S. Ruan, T. Day, J. Watmough, F. Brauer, P. Van den Driessche, D. Gabrielson, C. Bowman, M.E. Alexander, S. Ardal, J. Wu, B.M. Sahai, ~``Modelling strategies for controlling sars outbreaks,"~Proc. R. Soc. Lond. B~\textbf{271},~2223--2232~(2004).

\bibitem{Keeling08} M.J. Keeling, .P Rohani,~``Modeling Infectious Diseases in Humans and Animals,"~(Princeton University Press, London, 2008).

\bibitem{Lan20} L. Lan, D. Xu, G. Ye, C. Xi, S.  Wang, Y. Li, H. Xu,~``Positive RT-PCR test results in patients recovered from COVID-19,"~JAMA.~\textbf{323}(15),~1502--1503 (2020).

\bibitem{Liu20} Z. Liu, P. Magal, O. Seydi, G. Webb,~``Predicting the cumulative number of cases for the COVID-19 epidemic in China from early data,"~Math. BioSci. Eng.,~\textbf{17}(4),~3040--3051~(2020).

\bibitem{Diekmann90} O. Diekmann, J.S.P. Heesterbeek, J.A.J. Metz,~``On the definition and the computation of the basic reproduction ratio $R_0$ in models for infectious diseases in heterogeneous populations,"~J. Math. Biol.~\textbf{28},~365--382~(1990).

\bibitem{covid-19tr} India covid-19 tracker.~https://www.covid19india.org/~(2020), ~(Retrieved on May 30, 2020).

\bibitem{Banerjee15} S. Banerjee, S. Khajanchi, S. Chaudhuri,~``A Mathematical Model to Elucidate Brain Tumor Abrogation by Immunotherapy with T11 Target Structure,"~PLoS ONE~\textbf{10}(5),~e0123611.~doi.org/10.1371/journal.pone.0123611.

\bibitem{Corey20} L. Corey, J.R. Mascola, A.S. Fauci, F.S. Collins,~``A strategic approach to COVID-19 vaccine R\&D,"~Science,~\textbf{368}(6494),~948--950~(2020).

\bibitem{Lurie20} N. Lurie, M. Saville, R. Hatchett, J. Halton,~``Developing Covid-19 Vaccines at Pandemic Speed,"~N Engl J Med.~\textbf{382},~1969--1973~(2020).

\bibitem{Castillo04} C. Castillo-Chavez, B. Song,~``Dynamical models of tuberculosis and their applications,"~Math. Biosci. Eng.~\textbf{1}(2),~361--404~(2004).

\bibitem{Khajanchi18} S. Khajanchi, D.K. Das, T.K. Kar,~``Dynamics of tuberculosis transmission with exogenous reinfections and endogenous reactivation,"~Physica A.~\textbf{497},~52--71~(2018).












\end{thebibliography}
\end{document}